\documentclass[acmsmall,screen,nonacm]{acmart}

\setcopyright{none}
\renewcommand\footnotetextcopyrightpermission[1]{}

\makeatletter
\def\@authorsaddresses{}
\makeatother

\makeatletter
\AtBeginDocument{%
  \fancypagestyle{firstpagestyle}{%
    \fancyhf{}%
    \fancyfoot[C]{\thepage}%
  }%
}
\makeatother

\AtBeginDocument{%
  }

\usepackage{algorithmic}
\usepackage{graphicx}
\usepackage{textcomp}
\usepackage{xcolor}
\usepackage{url}
\usepackage{fancyhdr}
\usepackage{hyperref}
\usepackage{subcaption}
\usepackage{amsmath}
\usepackage{xspace}
\usepackage{multirow}
\usepackage{booktabs}
\usepackage{cleveref}
\usepackage{enumitem}
\usepackage{amsmath}
\usepackage{makecell}
\usepackage{textcomp}
\usepackage{textgreek}
\usepackage{pifont}
\usepackage{wrapfig}

\DeclareMathOperator*{\argmin}{argmin}

\usepackage{siunitx}
\DeclareSIUnit\dollar{\$}
\DeclareSIUnit{\month}{month}
\DeclareSIUnit{\thousand}{k}
\DeclareSIUnit{\million}{M}

\usepackage{pifont}
\newcommand{\circled}[1]{\ding{\numexpr#1+201}}
\newcommand{\circledtext}[1]{\raisebox{.5pt}{\textcircled{\raisebox{-0.9pt} {\small #1}}}}

\newtheorem*{theorem*}{Theorem}

\usepackage[skins]{tcolorbox}

\newcounter{takeawaycounter}
\crefname{takeawaycounter}{Takeaway}{Takeaways}

\newtcolorbox[use counter=takeawaycounter]{takeawaybox}[2]{
  enhanced, width=\hsize,left=1pt,right=1pt,top=1pt,bottom=1pt,colback=green!65!black!5!white,boxrule=1pt,colframe=green!65!black!25!white,
  label type=takeawaycounter,
  label = #1,
  title = \textcolor{black}{\bf Takeaway~\arabic{takeawaycounter}~#2}
}

\usepackage{soul}
\makeatletter
\def\SOUL@hlpreamble{%
\setul{}{2.2ex}
\let\SOUL@stcolor\SOUL@hlcolor
\SOUL@stpreamble
}
\makeatother
\soulregister\name7
\soulregister\cite7
\soulregister\ref7
\soulregister\Cref7
\soulregister\pageref7
\soulregister\eg7
\soulregister\ie7
\soulregister\texttt7
\soulregister\circled7
\soulregister\encircle7
\soulregister\uline7
\soulregister\circled7
\soulregister\circledtext7
\soulregister\todo7
\soulregister\henry1

\newcommand{\name}{GreenCache\xspace}
\newcommand{\cachename}{LCS\xspace} 

\newcommand{\eg}{e.g.,\xspace}
\newcommand{\ie}{i.e.,\xspace}

\newcommand{\unitCI}{gCO\textsubscript{2}e/kWh\xspace}
\newcommand{\unitem}{kgCO\textsubscript{2}e\xspace}

\title{Cache Your Prompt When It's Green --- Carbon-Aware Caching for Large Language Model Serving} 



\author{Yuyang Tian}
\email{yuyang.tian@uwaterloo.ca}
\affiliation{\institution{University of Waterloo}
\city{Waterloo}
\state{ON}
\country{Canada}
}

\author{Desen Sun}
\email{desen.sun@uwaterloo.ca}
\affiliation{\institution{University of Waterloo}
\city{Waterloo}
\state{ON}
\country{Canada}
}

\author{Yi Ding}
\email{yiding@purdue.edu}
\affiliation{\institution{Purdue University}
\city{West Lafayette}
\state{IN}
\country{United States}
}

\author{Sihang Liu}
\email{sihangliu@uwaterloo.ca}
\affiliation{\institution{University of Waterloo}
\city{Waterloo}
\state{ON}
\country{Canada}
}

\begin{document}
\begin{abstract}

As large language models (LLMs) become widely used, their environmental impact, especially the carbon emission problem, has attracted more attention. Prior studies focus on compute-related carbon emissions. In this paper, we find that storage is another key contributor. LLM caching, which saves and reuses KV caches for repeated context,
reduces operational carbon by avoiding redundant computation. However, this benefit comes at the cost of embodied carbon from high-capacity, high-speed SSDs. As LLMs scale, the embodied carbon of storage grows significantly.
To address this tradeoff, we present \name{}, a carbon-aware cache management framework that dynamically derives resource allocation plans for LLM serving. \name{} analyzes the tradeoff between carbon emissions and SLO satisfaction, reconfiguring the resource over time to maintain this balance under dynamic workloads. Evaluations from real traces demonstrate that \name{} achieves an average carbon reduction of 15.1\,\% when serving Llama-3 70B in the FR grid, with reductions reaching up to 25.3\,\%,  while staying within latency constraints for > 90\,\% of requests.

\end{abstract}

\maketitle
\thispagestyle{empty}


\section{Introduction}
Large language models (LLMs) become widely adopted across applications---from dialogue systems~\cite{chatgpt, deepseek, Gemini} to healthcare \cite{Peng2023} and developer tools \cite{copilot}. Although LLMs demonstrate remarkable capabilities, they raise concerns about the environmental impact~\cite{ding2024sustainable,stojkovic2025dynamollm,faiz2024llmcarbon,li2025ecoserve,greenLLM,greenllm_ieee_cal,clover,li2024towards,samsi2023fromwordstowatt}, particularly in terms of carbon emissions measured in carbon dioxide equivalent (CO\textsubscript{2}e).

Carbon emissions in LLM serving systems come from two main sources: operational carbon and embodied carbon. 
Operational carbon emissions come from the electricity consumed by computing.
The amount of carbon emitted per kilowatt-hour (kWh) depends on the grid's energy source mix, quantified by the \emph{carbon intensity} (CI), measured in the unit of \unitCI{}~\cite{maji2022carboncast}. Renewable energy sources like solar and hydro have low carbon intensities, while non-renewable sources like coal and gas have high carbon intensities. Embodied carbon emissions come from the hardware manufacturing process~\cite{wu2022sustainable,ACT}. Serving LLMs at scale is compute-intensive. It demands high power that will likely lead to high operational carbon emissions, and at the same time, requires high-performance hardware like high-end GPUs and machine learning accelerators that have high embodied carbon. For example, serving a single LLM prompt can emit over 20 times more carbon than a conventional web search query \cite{whyyourinternet2020,chatgptcarbon2023}. 

To reduce the carbon emissions of LLM serving, prior work has focused primarily on optimizing the compute side (e.g., GPUs). For example, DynamoLLM~\cite{stojkovic2025dynamollm} improves energy efficiency to lower operational emissions; GreenLLM~\cite{greenLLM, greenllm_ieee_cal} and EcoServe~\cite{li2025ecoserve} reuse older GPUs to lower embodied emissions. However, these efforts largely overlook another major contributor to carbon impact in LLM serving systems: storage.

While compute-related carbon emissions have been widely studied, the carbon emissions of storage remain an underexplored but critical component in sustainable LLM deployment. LLM tasks often involve a long context to provide specific information for better generation quality, such as long chat histories in multi-turn conversations and long documents in comprehension tasks. Processing these long inputs can require thousands of tokens before generation \cite{cachegen}. To address this, prior work \cite{cachedattention, cachegen, hcache, yao2024cacheblend} has introduced context caching, which saves the KV cache of contexts in fast storage and reuses them when needed again. This reduces redundant computation, lowers latency, and decreases resource usage, ultimately reducing operational carbon emissions.

Despite the benefits, context caching requires high-capacity and high-speed storage such as SSDs, which come with embodied carbon emissions. Prior work \cite{tannu2023diretysecretssd,ke2024improving,seagate_embodied,bhagavathula2024understanding,mcallister2024call,mcallister2025fairywren,lca-dell-r740} has shown that SSDs contribute to over 75\,\% of datacenter servers' total embodied carbon emissions \cite{tannu2023diretysecretssd,lca-dell-r740}. This means that caching does not come for free, but can significantly increase the overall carbon emissions of LLM serving. 
In a cloud environment, storage is provisioned on demand and can scale. 
Therefore, the embodied carbon from service is counted by the storage size and its time of usage, like prior studies on embodied carbon of hardware \cite{ACT,han2025fair}.
However, there is no effective characterization or optimization strategy that optimizes storage carbon for LLM serving.

In this work, we present \name{}, a carbon-aware caching framework to mitigate storage-related carbon emissions from caching in LLM serving systems. 
Balancing the tradeoff between increased embodied carbon and reduced operational carbon makes identifying the optimal caching configuration particularly challenging. 
To guide cache decisions, we characterize LLM serving caching behavior and make two key observations.
First, higher loads amplify both latency reduction and carbon savings from caching, as cache hits substantially shorten the prefill phase.
Second, carbon emission savings are strongly influenced by the carbon intensity (CI). 
When the carbon intensity becomes high, operational carbon from computation dominates; conversely, when it is low, embodied carbon from cache storage (SSD in this study) becomes the primary contributor. 
Both the load and CI fluctuate in a realistic environment, as load varies with user activity~\cite{stojkovic2025dynamollm}, and CI changes with real-time energy sources~\cite{maji2022carboncast}.
These insights reveal that the optimal cache configuration is jointly determined by these dynamic factors --- the load, which governs operational carbon savings through latency reduction, and the carbon intensity, which affects the net carbon savings.
Accordingly, \name{} adapts to both load and CI conditions to continuously optimize cache configuration for optimal carbon efficiency.

\name{} first profiles the performance and power under various workloads and cache sizes.
In practice, \name{} predicts both carbon intensity and workload based on historical values, enabling adaptive cache reconfiguration that anticipates future dynamics up to 24 hours in advance to preserve enough time to allow sufficient warm-up.
However, finding the most carbon-efficient cache size is not sufficient. 
LLM serving applications are typically required to meet Service Level Objectives (SLOs) for both Time To First Token (TTFT) and Time Per Output Token (TPOT), with a target attainment rate (\eg 90\,\%) to ensure service quality~\cite{distserve,stojkovic2025dynamollm,moocacke2024arxiv}.
To ensure SLO attainment while achieving carbon efficiency, we model this optimization as an Integer Linear Programming (ILP) problem. 
\name{} integrates an ILP solver that takes the predicted carbon intensity and load, current TTFT and TPOT, and the SLOs as input, and determines the most carbon-efficient cache configuration that meets the SLO target, effectively balancing the operational carbon savings (from reduced computation) and embodied carbon emissions (from storage). 

In addition, we introduce a new cache replacement policy for \name{}, as conventional cache replacement policies such as Least Recently Used (LRU) do not optimize for carbon emissions. 
We find that longer reused context (\ie more tokens upon a cache hit) yields more computation savings, reducing the operational carbon.
On the other hand, the cache entry size varies, leading to different embodied carbon overheads. 
Therefore, we design a new policy, Least Carbon Savings (\cachename{}), that incorporates not only cache access frequency and recency, but also the embodied carbon costs and operational carbon savings from reused context.

We implement \name{} on top of an open-source context caching system, LMCache~\cite{lmcache}. 
\name{} leverages input SLO constraints and real-time carbon intensity to minimize total carbon emissions in serving LLMs while meeting the performance SLO attainment goal. 
We adapt two LLMs on top of the \name{} framework: Llama-3 70B and 8B models \cite{llama3}, and evaluate on two tasks: multi-turn conversation using a ShareGPT dataset \cite{sharegpt_dataset} and document reading comprehension using the TriviaQA dataset \cite{joshi2017triviaqa}. 
For a realistic evaluation, we take 24-hour LLM request rates from the Azure dataset~\cite{azurellm2024} and carbon intensities from the CarbonCast dataset~\cite{maji2022carboncast}.
Experiments are conducted on a system with 4$\times$ NVIDIA L40 GPUs and configurable SSD storage of up to 16 TB. 

The contributions of this paper are as follows:
\begin{itemize}[leftmargin=*]
    \item To the best of our knowledge, this is the first work that systematically studies the embodied carbon emissions of storage due to caching in LLM serving systems. 
    \item We design and implement a caching framework for LLM serving that dynamically reconfigures cache size in response to fluctuating CIs and workloads, while meeting the SLO attainment goal. 
    \item Our evaluation based on realistic request rate and CI traces demonstrates that \name{} reduces average carbon emissions by 15.1\,\% in low-CI grids such as FR, and can still achieve up to a 6.91\,\% reduction in higher-CI grids like CISO when serving the Llama-3 70B model. 
    \item  The source code of this work is available at \url{https://greencache.persistentmemory.org}.

\end{itemize}
\section{Background}

In this section, we first introduce the LLM serving process. We then discuss the cache mechanism for LLM serving for performance optimization. Finally, we discuss carbon accounting methods for operational and embodied emissions.

\subsection{LLM Serving Process}

LLM serving, typically based on the Transformer architecture, operates by predicting the output in a sequence given a prompt. 
LLM serving has two phases: \emph{prefill} and \emph{decode}.
The prefill phase (or prompt phase) processes the input prompt and generates the first output token. Its performance metric is measured by \emph{Time To First Token (TTFT)} --- the time from receiving the prompt to producing the first token.
The decode phase (or token generation phase) generates tokens one by one in an autoregressive manner. Its performance metric is measured by \emph{Time Per Output Token (TPOT)} --- the time between each output token. 

While attention operations are computationally intensive, most operations with Key and Value remain the same across decode iterations, providing an opportunity to store and reuse these tokens, referred to as the \emph{KV cache} \cite{pagedattention}. With \emph{KV cache}, the decode phase becomes memory-intensive rather than compute-intensive, and increasing the batch size can significantly improve throughput. 

In LLM serving systems, batching is a key technique for maximizing throughput \cite{pagedattention,orca,distserve}. Naively batching all existing requests and disallowing batch enlargement mid-processing results in suboptimal throughput.
To address this, recent work introduces continuous batching \cite{orca}, which allows new requests to be inserted during the decode phase, significantly improving GPU efficiency.

\subsection{Caching for LLM Serving}
\label{subsec:caching-background}

To further reduce the heavy computation in the prefill phase, recent work has introduced caching techniques~\cite{cachedattention, cachegen, cheng2024large, yao2024cacheblend, pagedattention, hcache}. These studies indicate that many requests share overlapping tokens. 
For example, in multi-turn conversations, later turns often include the entire chat history. In reading comprehension, multiple questions may ask about the same document. Caching enables reusing the KV cache from previous requests, avoiding redundant computation for repeated context and trading off extra storage for better performance.

\begin{figure}[t]
  \centering
    \subfloat[Caching procedure.]{\includegraphics[width=0.48\linewidth]{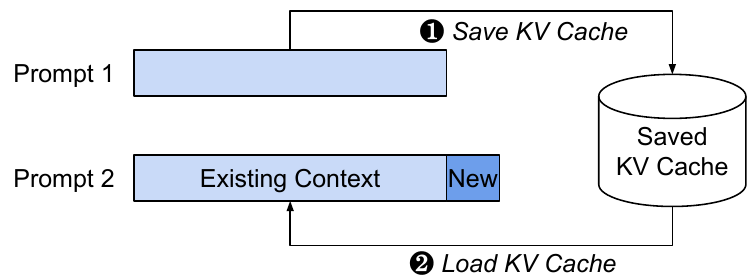}
    \label{fig:prefix_cache}}
    \hfill
    \subfloat[Timeline of no cache and a cache hit.]{\includegraphics[width=0.48\linewidth]{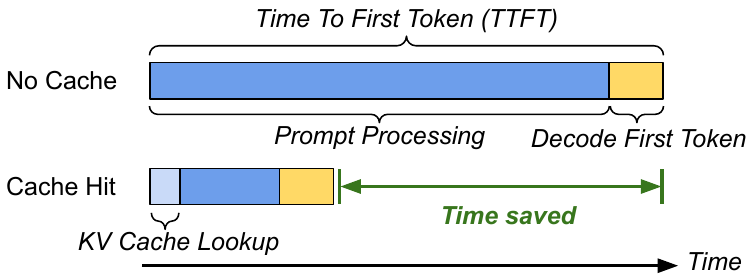}
    \label{fig:cache_timeline}}
  \caption{\label{fig:cached_attention} Illustration of caching for LLM serving.}
  \Description{}
\end{figure}

\Cref{fig:prefix_cache} illustrates a system that reuses the KV cache of the context. 
Step~\circled{1}: upon a new request Prompt~1, the system processes it and saves its KV cache to storage. 
Step~\circled{2}: when the next request Prompt~2 arrives, the system looks up the cache for the existing context. When hit, the system stitches the KV cache of the cached tokens and the new tokens, and processes them together. 
\Cref{fig:cache_timeline} depicts the timeline with or without cache. 
Due to the intensive attention computation, KV cache lookup incurs a much lower overhead than processing. For example, the average TTFT in executing prompts in the ShareGPT dataset \cite{sharegpt_dataset} using the Llama-3 70B model on $4 \times$ L40 GPUs is 1.7 seconds, but loading the KV cache of previous context only takes 0.03 seconds on average. 
Therefore, cache reduces TTFT significantly by eliminating redundant processing in the prefill phase. 
Even though caching does not reduce computation in the decode phase, it reduces the waiting time for the decode phase in continuous batching, improving the decode latency. 
Despite the performance gains, an effective caching system fundamentally relies on the storage backend to maintain a high volume of KV cache data that can be accessed at high speed. 
For example, caching a 1000-token context for 1 million prompts of a Llama-3 70B model takes more than 300 TB \cite{lmcache_calculator}.

\subsection{Carbon Emission Accounting}
\label{subsec:carbon-background}

\begin{figure}
    \centering
    \includegraphics[width=\linewidth]{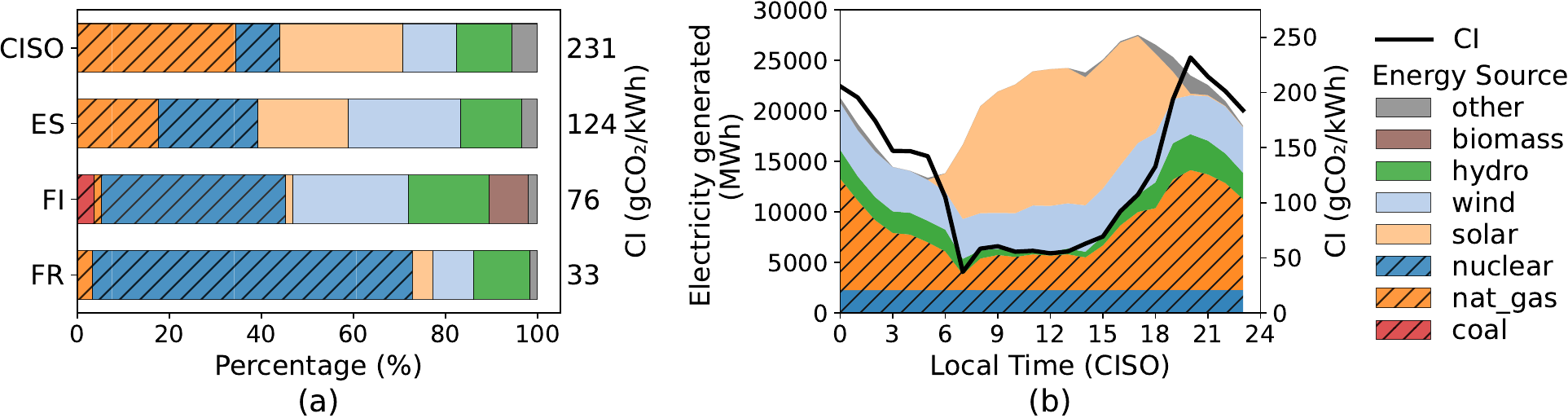}
    \caption{(a) Average carbon intensity (CI) and energy sources of four grids in 2024 \cite{ele_maps}. (b) CI variation due to energy sources of the CISO grid on July 6, 2022 \cite{maji2022carboncast}.}
    \label{fig:ci-decomposition}
    \Description{}
\end{figure}

LLM serving is highly compute-intensive and results in significant environmental impact. According to the carbon modeling tool from prior work~\cite{ACT,Sustainable_hpc,clover,faiz2024llmcarbon,greenLLM,greenllm_ieee_cal,nguyen2024towards,li2024towards}, the total carbon emissions of serving LLMs on a hardware platform are the sum of the operational ($C_{\rm o}$) and embodied ($C_{\rm e}$) carbon emissions. 
As the cloud dynamically allocates resources, a user's embodied carbon emissions are attributed to the duration for which the computing platform is allocated. 
Therefore, the embodied carbon emissions are proportional to the execution time $\texttt{T}$, amortized over the platform's lifetime ($\texttt{LT}$). 
For example, the typical lifetime of hardware components is around 5 years~\cite{ACT,ostrouchov2020gpulife}, which we follow in the rest of this paper unless specified. 
We calculate the total carbon emissions as 
\begin{align}\label{eq:carbon-total}
C = C_{\rm o} + \frac{\texttt{T}}{\texttt{LT}} C_{\rm e}.
\end{align}

\begin{table}[t]
\small
    \caption{Embodied carbon of major server components.}
    \label{tab:embodied}
    \centering
    \begin{tabular}{lll}
    \toprule
    \textbf{Component} & \textbf{Type} & \textbf{Embodied Carbon} \\  \midrule
       CPU  & AMD 7453 & 9.3 kgCO\textsubscript{2}e \cite{ACT} \\
       GPU  & 4 \texttimes{} NVIDIA L40 & 106.4 kgCO\textsubscript{2}e \cite{ACT,nguyen2024towards} \\
       Memory & 512 GB DDR4 & 30.8 kgCO\textsubscript{2}e \cite{ACT} \\
       Storage & (Up to) 16 TB SSD & (Up to) 480 kgCO\textsubscript{2}e \cite{ACT} \\
    \bottomrule
    \end{tabular}
\end{table}

The \emph{operational carbon} ($C_{\rm o}$) is computed as the product of the energy $E$ consumed by running LLMs on the target hardware and the carbon intensity ($\texttt{CI}$) during the execution
\begin{align}\label{eq:carbon-op}
  C_{\rm o} = E \times \texttt{CI}, 
\end{align}
where \texttt{CI} measures the amount of carbon dioxide equivalent emitted per unit of electricity generation (i.e., \unitCI{}). A lower carbon intensity value indicates a higher proportion of renewable energy used in electricity generation from the power grid.
\Cref{fig:ci-decomposition}a shows four grids: FR (France), FI (Finland), ES (Spain), and CISO (California, USA), which differ in their compositions of energy sources.
\Cref{fig:ci-decomposition}b shows the electricity sources vary in the CISO grid throughout a day.

The \emph{embodied carbon} ($C_{\rm e}$) from hardware comes primarily from the semiconductor manufacturing process. 
In this work, we focus on the main hardware components that contribute to the majority of embodied carbon for LLM serving, including GPUs, CPUs, memory, and storage.
\Cref{tab:embodied} lists the hardware components and their embodied carbon, which is modeled based on their processor chip area and memory capacity~\cite{ACT}. 
We calculate the total embodied carbon emissions of the system running LLM service as
\begin{align}\label{eq:carbon-eb}
	C_{\rm e} = \hspace{-4mm}\sum_{{\rm comp}\in{\rm Platform}}\hspace{-4mm}C_{\rm e, comp} = C_{\rm e, GPU} + C_{\rm e, CPU} + C_{\rm e, Mem} + C_{\rm e, SSD}.
\end{align} 
While prior studies have examined the embodied carbon of compute components like CPUs and GPUs, the impact of storage, particularly SSDs, for LLM serving remains underexplored, despite evidence that it contributes significantly to embodied emissions~\cite{Sustainable_hpc,tannu2023diretysecretssd,mcallister2024call,bhagavathula2024understanding,lca-dell-r740,seagate_embodied}.
In our platform, the embodied carbon from SSDs consumes 76.6\,\% of the total embodied carbon of the server, similar to the over 75\,\% fraction reported by a prior study \cite{tannu2023diretysecretssd}.

\section{Performance and Carbon Emission of LLM Caching}

In this section, we analyze the performance and carbon emissions of caching in LLM serving and summarize the observations as takeaways.  
We illustrate the results using the Llama~3 70B model~\cite{llama3} and the ShareGPT dataset~\cite{sharegpt_dataset}, evaluated on the server described in~\Cref{tab:embodied}. We use LMCache~\cite{lmcache} as the caching system, which saves context on 16 TB SSDs. 
The cache has been initialized with 200k prompts, and we record the performance and carbon emissions of 500 prompts. 
Details of our evaluation methodology are provided in \Cref{subsec:methodology}.

\subsection{Performance of LLM Caching }

We first study the performance of context caching in LLM serving by focusing on two factors: context length and request rate. 
Since the prefill phase accounts for only a small fraction of the total latency, we present all latency results in log scale to better visualize the prefill latency.

\begin{figure}
\begin{minipage}[t]{0.51\linewidth}
    \centering
    \includegraphics[width=\linewidth]{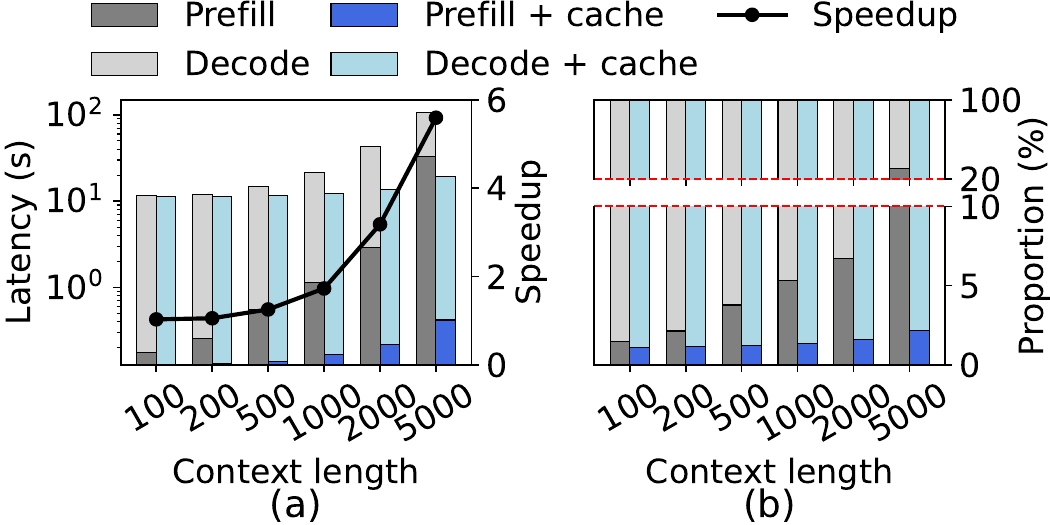}
    \caption{(a) Latency and speedup from caching under different context lengths. (b) Latency breakdown. }
    \label{fig:context_length}
    \Description{}
\end{minipage}
\hfill
\begin{minipage}[t]{0.45\linewidth}
    \centering
    \includegraphics[width=1\linewidth]{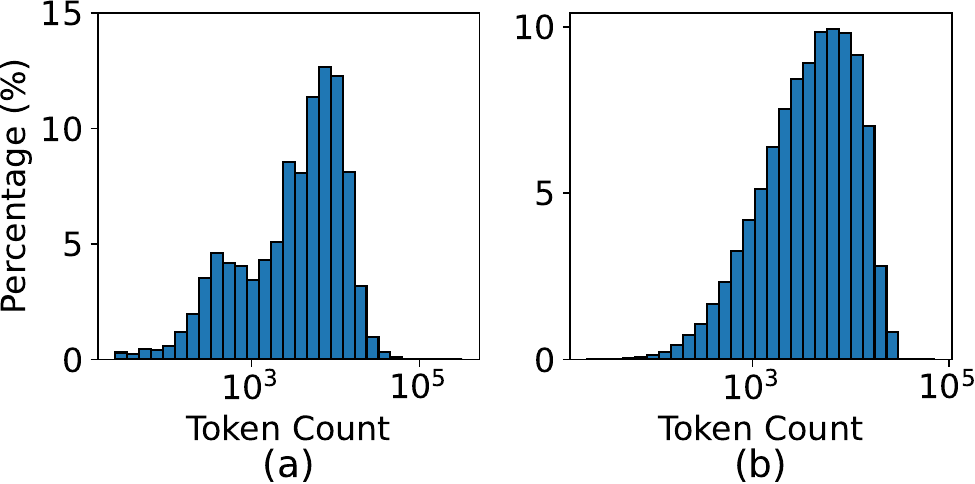}
    \caption{Distribution of context length in (a) ShareGPT \cite{sharegpt_dataset} and (b) TriviaQA \cite{joshi2017triviaqa}. }
    \label{fig:cdf-sharegpt}
    \Description{}
\end{minipage}
\end{figure}

\subsubsection{Impact of Context Length}
\label{subsubsec:context_length}

Caching reduces the computation of the existing context. Therefore, we first evaluate its benefits under different context lengths.
\Cref{fig:context_length}a shows the prefill and decode latency breakdown and speedup from caching under variable context lengths.
The speedup increases as the cached context gets longer.
Furthermore, caching brings more improvements for the prefill phase than the decode phase, as the fraction in \Cref{fig:context_length}b shows. This is because prefill directly benefits from cached context, but decode indirectly benefits from reduced waiting time. 

We further analyze the context length in two common LLM tasks where caching can be useful.
First, we study the context length of a multi-turn conversation task using a ShareGPT dataset~\cite{sharegpt,sharegpt_dataset}. Each conversation contains a number of turns between the user and the LLM chatbot. 
We observe that the context length varies, as \Cref{fig:cdf-sharegpt} shows that 77.2\,\% of prompts have over 1000 tokens from prior conversation turns. 
Second, we study another task, document reading comprehension, using the TriviaQA dataset \cite{joshi2017triviaqa}. 
Unlike multi-turn conversation, its context is the whole document that corresponds to the user's question, with an average context length of 5880 tokens. 
Therefore, the majority of requests benefit from context caching.

\begin{takeawaybox}{box:length-perf}{}
Longer context lengths yield more benefits from caching, as more redundant prefill computation can be eliminated. However, context length varies in tasks and prompts, leading to variable benefits. Therefore, longer contexts should be prioritized in caching. 
\end{takeawaybox}

\subsubsection{Impact of Request Rate}

The request rate is a key factor that impacts the performance benefits from caching. 
We conduct an experiment to evaluate the latency under different request rates, as \Cref{fig:rate-time} shows. 
As the request rate becomes higher, the average latency reduction of prefill increases. 
At the same time, the average decoding latency also reduces and benefits more when the request rate is higher. 
This is because reduced prefill latency also saves waiting time for decoding, as discussed in \Cref{subsec:caching-background}.

\begin{takeawaybox}{box:rate-perf}{}
Higher request rates benefit more from caching, as the computation reduction by cached context is more prominent when the request rate is high. Therefore, more context should be cached when the system is under high load. 
\end{takeawaybox}

\begin{figure}
\begin{minipage}[t]{0.48\linewidth}
    \centering
    \includegraphics[width=\linewidth]{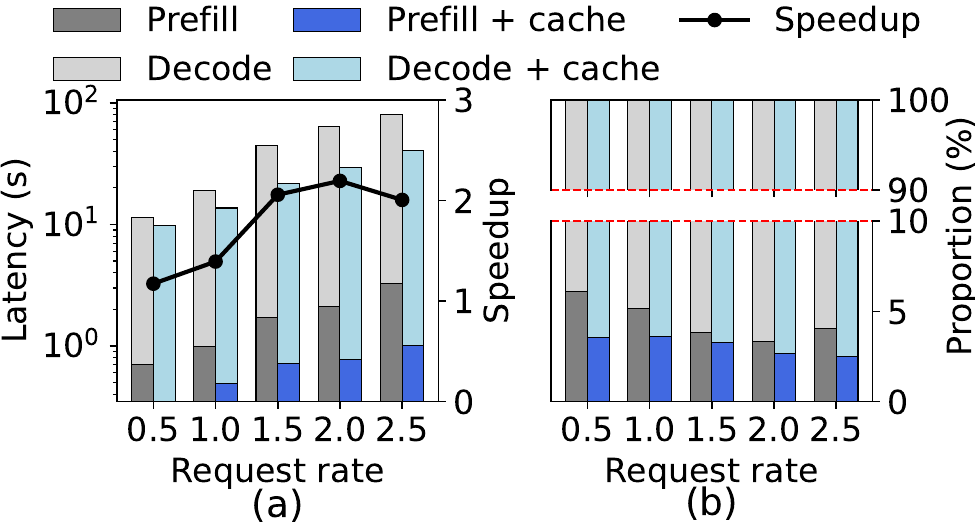}
    \caption{(a) Latency of prefill and decode under different request rates, and the speedup from caching. (b)~Fraction of prefill and decode latency.}
    \label{fig:rate-time}
    \Description{}
\end{minipage}
\hfill
\begin{minipage}[t]{0.48\linewidth}
    \centering
    \includegraphics[width=\linewidth]{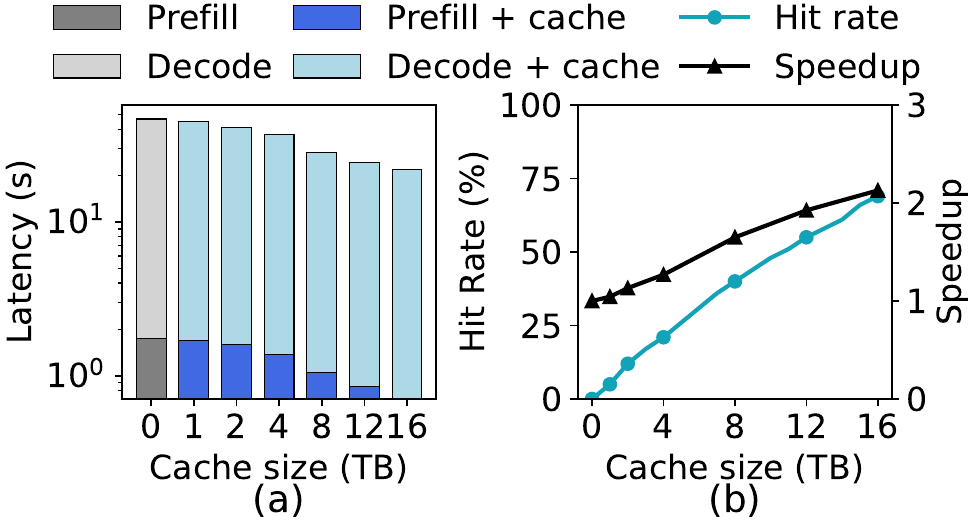}
    \caption{(a) Latency and speedup from caching under different cache sizes. (b) Cache hit rate.}
    \label{fig:size-time}
    \Description{}
\end{minipage}
\end{figure}

\subsubsection{Impact of Cache Size}
\label{subsubsec:cache_size}

The cache size is another key factor for performance benefits. 
We set a fixed request rate of 1.5 prompts/s and vary the cache sizes from 1 to 16 TB. Then, we compare the performance with the no-cache configuration. 
\Cref{fig:size-time} shows that larger cache sizes yield higher speedup and lower latency, as more prompts hit and reuse the KV cache of the cached context. 

\begin{takeawaybox}{box:size-perf}{}
Larger cache sizes improve performance by increasing hit rates, allowing more prompts to reuse KV caches. However, the benefit may not always scale linearly, as the hit rate improvement slows down once the cache reaches a certain size.
\end{takeawaybox}

\subsection{Carbon Emissions of LLM Caching}

In this subsection, we model and characterize the carbon emissions of LLM caching systems.

\subsubsection{Cache Carbon Modeling} \label{subsubsec:cache-carbon-modeling}

\Cref{subsec:carbon-background} introduces the standard carbon modeling method for computing systems. 
Building upon it, we introduce a carbon modeling approach for LLM caching. 
In this work, we use SSD, the most commonly used storage in LLM serving systems \cite{hcache, lmcache,yao2024cacheblend, cachedattention,cachegen}, as the cache hardware.\footnote{The same carbon modeling also applies to other caching mediums, such as DRAM, CXL-attached memory, and HDD. } 
In a cloud environment, storage is provisioned on demand and can be flexibly resized. 
In LLM context caching, the cache size reflects the number of cached tokens, as the service acquires more storage when the number of cached tokens increases and releases storage when the demand for cached tokens decreases.
We model the cache’s embodied carbon as proportional to the current allocated storage size $\rm S_{Alloc}$, reflecting a cloud scenario where only the reserved storage contributes to embodied carbon, rather than the entire device.
Let $C_{\rm e, SSD}^{\rm Unit}$ denote the embodied carbon per unit of storage (e.g., \unitem{} per TB).
The embodied carbon attributed to the allocated SSD is then scaled by the fraction of its lifetime during which the storage is occupied.
In our evaluation system, we limit the SSD size to 16 TB, as shown in \Cref{tab:embodied}.
Accordingly, we define the carbon emissions of the cache as follows:
\begin{align}\label{eq:ssd-embodied}
C_{\rm e, Cache} = {\rm S_{Alloc}} \times \frac{\texttt{T}}{\texttt{LT}} C_{\rm e, SSD}^{\rm Unit}.
\end{align}
For operational carbon emissions, we first measure the power of every server component when serving LLM and the prompt execution latency to obtain energy consumption. 
Then, we multiply it by the carbon intensity of the grid, as introduced in \Cref{eq:carbon-op}.
Together, the total carbon emissions can be calculated using \Cref{eq:carbon-total}, expressed as:
\begin{align}\label{eq:cache+other}
C = E\times\texttt{CI}+C_{\rm e,Cache} + \frac{\texttt{T}}{\texttt{LT}}C_{\rm e, Others}. 
\end{align}
This equation indicates a tradeoff between embodied and operational carbon emissions.
Caching improves LLM serving performance by lowering the execution time and energy consumption, which reduces operational carbon emissions (the $E\times\texttt{CI}$ component).
However, caching that uses extra storage incurs embodied carbon (the $C_{\rm e,Cache}$ component). Next, we examine how different serving conditions impact the carbon emissions of LLM caching.

\subsubsection{Carbon Emissions of LLM Caching.}
\label{subsubsec:carbon-motivation}

\begin{figure}
\begin{subfigure}[b]{0.41\linewidth}
    \centering
    \includegraphics[trim={0 0 2.6in 0}, clip, width=\linewidth]{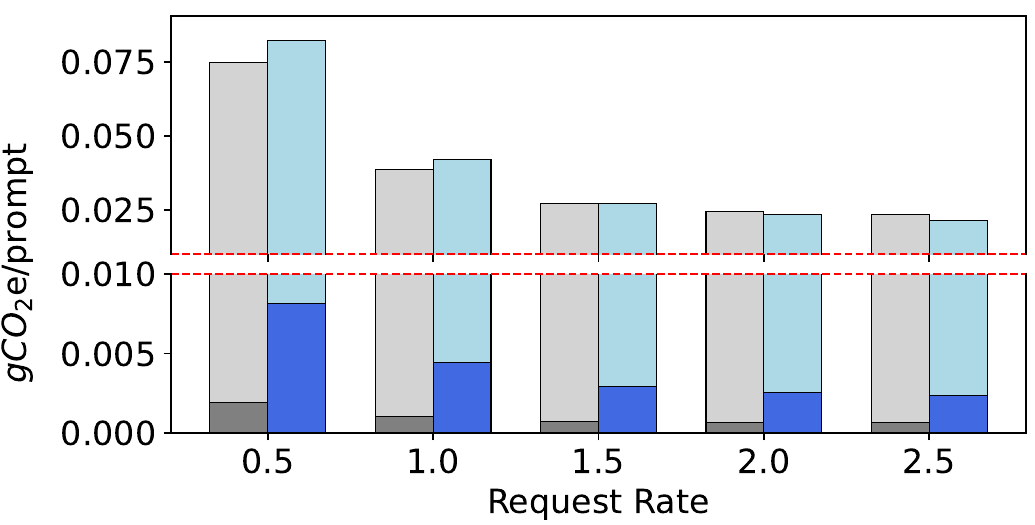}
    \caption{}\label{fig:carbon-with-qps}
    \Description{}
\end{subfigure}
\hfill
\begin{subfigure}[b]{0.56\linewidth}
    \centering
    \includegraphics[width=\linewidth]{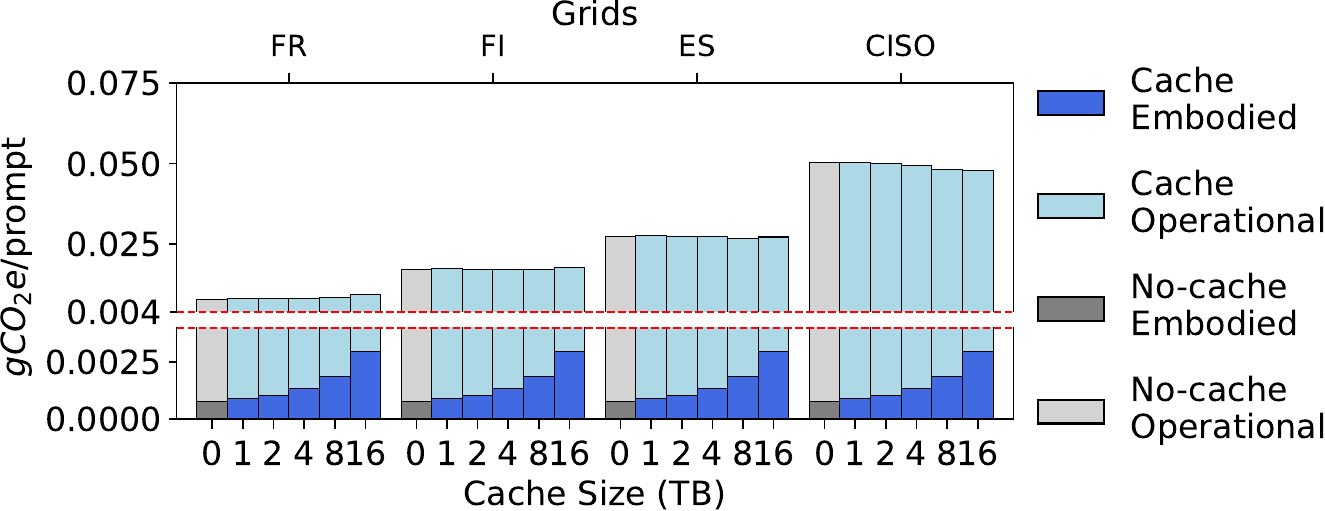}
    \caption{}\label{fig:carbon-with-ci}
    \Description{}
\end{subfigure}
\caption{Carbon emissions per request (a) in the ES grid under different request rates, and (b) under different cache sizes and grids (based on grid average CI in \Cref{fig:ci-decomposition}a). }
\end{figure}

We first measure the caching system’s carbon emissions under varying request rates. 
\Cref{fig:carbon-with-qps} demonstrates that the average per-prompt carbon emissions vary with the request rate.
All carbon emissions are calculated with ES grid's carbon intensity. 
At low request rates, caching incurs higher emissions due to the embodied carbon of SSDs. As the request rate increases, operational carbon savings from caching become dominant, while the relative impact of embodied carbon diminishes.

\begin{takeawaybox}{box:rate-carbon}{}
The carbon emission savings depend on the request rate. At higher loads, caching yields greater carbon reductions by effectively mitigating the elevated operational carbon. 
\end{takeawaybox}

\begin{figure}
    \centering
    \includegraphics[width=\linewidth]{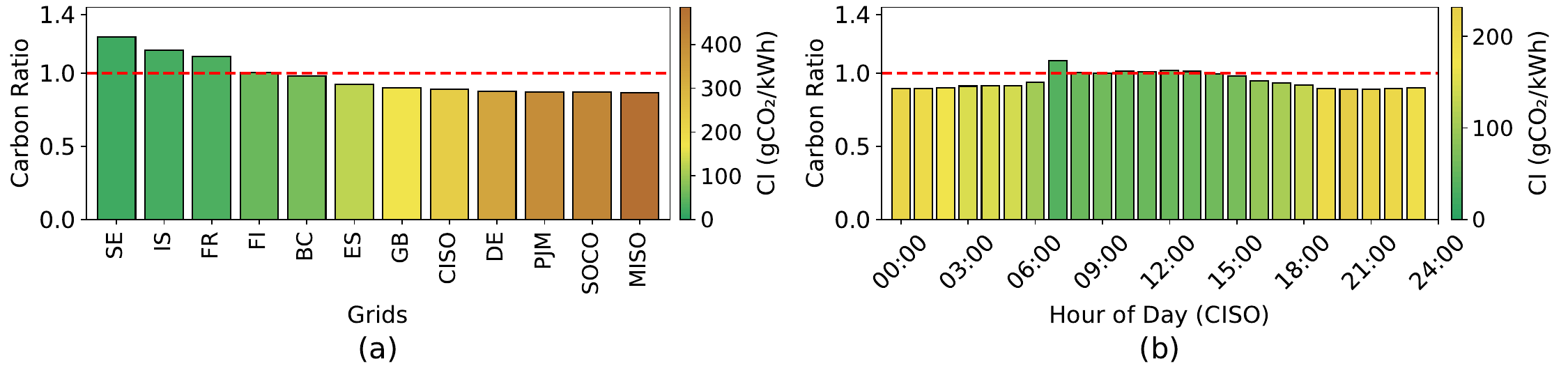}
    \caption{Carbon emission savings from caching in 12 grids. A ratio < 1 indicates carbon emission reduction.}
    \label{fig:ci-ratio-25}
    \Description{}
\end{figure}

We next study the impact of carbon intensity.
Under a request rate of 1.5 prompts/s, we calculate the total carbon emission savings using the average carbon intensity (CI) from 12 grids \cite{ele_maps}, as shown in \Cref{fig:ci-ratio-25}a.
Grids with a higher CI (on the right side) benefit more from caching. 
For instance, a 16 TB cache reduces carbon emissions by 7.5\,\% in MISO, where the CI reaches 485~\unitCI{}.
In contrast, low-CI grids (on the left side) benefit less or even incur higher emissions from caching. For example, in FR, where the CI is only 33~\unitCI{}, the same cache increases carbon emissions by 16.5\,\%.
As CI fluctuates within the same grid due to changing energy sources, we next evaluate the carbon emission savings of a 16 TB cache over a day in the CISO grid, as shown in \Cref{fig:ci-ratio-25}b.
The CI drops to its daily minimum of 37~\unitCI{} at 7 AM when the grid is mainly powered by renewable energy sources (as shown in \Cref{fig:ci-decomposition}b), making caching not carbon-efficient.
Conversely, at 8 PM, CI peaks at 232~\unitCI{}, and caching achieves its maximum carbon reduction.

The shift in carbon emission savings results from changes in the primary carbon contributor within the LLM serving system.
Consequently, when CI is high, a larger cache can save more computations and significantly reduce operational carbon, reducing the total carbon per prompt. 
Conversely, at low CI, embodied carbon dominates the total carbon emissions. Thus, increasing the cache size may negatively impact the carbon emissions.

To better understand this tradeoff between embodied carbon and operational carbon, we further evaluate the carbon emission savings under cache sizes ranging from 1 to 16 TB with the average CI of four grids in \Cref{fig:ci-decomposition}, as shown in \Cref{fig:carbon-with-ci}.
In all grids, the fraction of embodied carbon rises with cache size, while operational carbon reduction is more significant in high-CI grids.  
These results indicate that a fixed cache configuration is suboptimal under variable CIs.

\begin{takeawaybox}{box:ci-carbon}{}
Carbon emissions depend on carbon intensities and cache sizes. A higher carbon intensity can lead to lower carbon emissions with a larger cache, as the cache saves operational carbon due to computation reduction. 
When the carbon intensity is low, a large cache can increase the carbon emissions due to its embodied carbon.
Because the carbon intensity is highly dynamic, both across grids and within a grid, an adaptive method is needed to achieve optimal carbon savings. 
\end{takeawaybox}

\section{High-level Ideas of \name{}}

\name{} is a carbon-aware caching framework for LLM serving by making tradeoffs between the operational carbon reduction from caching and its extra embodied carbon. Next, we introduce two main ideas that enable adaptive caching and guarantee SLO attainment behind it. 

\subsection{Carbon-aware Adaptive Caching}

\Cref{box:rate-carbon,box:ci-carbon} reveal that the carbon savings from caching depend on both request rate and carbon intensity.
These two factors affect the choice of cache size configuration dynamically, as they depend on the task, time, and the grid.
Therefore, the first challenge is to find the optimal cache size given the dynamic factors. 

\name{} resizes the cache according to the request rate and carbon intensity during runtime.
The high-level idea is to first profile the LLM serving system's performance and power under different combinations of request rates and carbon intensity levels. 
When serving prompts, \name{} can find the optimal configuration of cache size that leads to the lowest per-prompt carbon emissions.
However, one challenge is that cache resizing, especially enlarging the cache, takes time to warm up, and thus the performance highly depends on the future conditions. 
To overcome this challenge, \name{} predicts the request rate and carbon intensity based on historical data, enabling optimal decision-making that considers not only the current condition but also future trends. 
\Cref{fig:adaptive_mode}a illustrates an example where \name{} adaptively increases the cache size when the predicted
load (\ie request rate) increases to better reduce the operational carbon. 
\Cref{fig:adaptive_mode}b shows another example where \name{} adaptively shrinks the cache size when the carbon intensity is predicted to reduce, to save the embodied carbon of cache storage.

\begin{figure}
\centering
\begin{minipage}[b]{0.28\linewidth}
    \centering
    \includegraphics[width=1\linewidth]{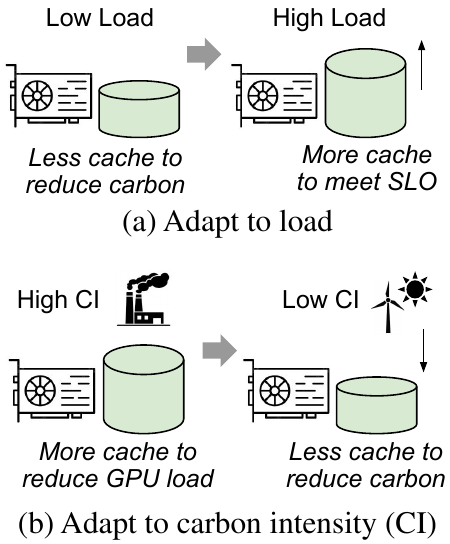}
    \caption{Adaptive caching.}
    \label{fig:adaptive_mode}
    \Description{}
\end{minipage}
\hfill
\begin{minipage}[b]{0.6\linewidth}
    \centering
    \includegraphics[width=1\linewidth]{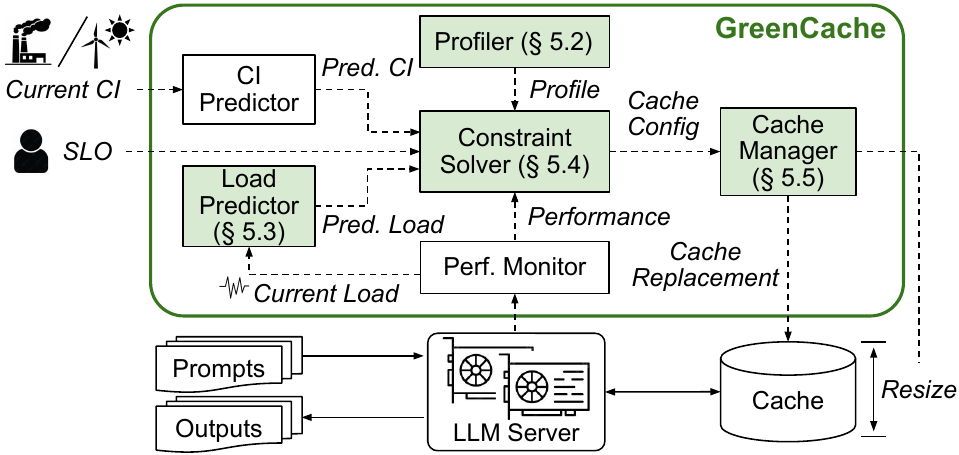}
    \caption{System overview of \name{}. Components in green are newly designed by this work.}
    \label{fig:overview}
    \Description{}
\end{minipage}

\end{figure}

\subsection{Performance SLO Attainment}
Only optimizing for carbon emissions can also negatively impact performance. 
As \Cref{box:size-perf} has shown, a small cache can lead to the minimum carbon emissions under low carbon intensity, but significantly degrade performance.
In LLM serving systems, Service Level Objectives (SLOs) typically define latency constraints for TTFT and TPOT, with the goal of ensuring that a high percentage of requests (e.g., 90\,\%) adhere to these thresholds.
Cache configurations that fail to achieve the required SLO attainment are not acceptable, even if they minimize carbon emissions.
Thus, the second challenge lies in minimizing carbon emissions while ensuring SLO attainment.

\name{} uses a constraint solver to find such a cache configuration.
We formulate the constrained carbon emission optimization as an Integer Linear Programming (ILP) problem.
The solver has an objective function that minimizes the carbon emissions based on the profile discussed in the first high-level idea.
At the same time, it enforces constraints to ensure SLO attainment.
In scenarios where a small cache is preferred to minimize carbon but can miss most SLOs, the solver guides \name{} to choose a larger cache that achieves targeted SLO compliance.

\section{\name{} Framework}

We first present an overview of \name{} and then describe its components.

\subsection{System Overview} \label{subsec:overview}

We design and implement \name{}, a carbon-aware caching framework for LLM serving.
\Cref{fig:overview} presents the system overview of \name{}. \name{} consists of six major components. 
First, a \emph{profiler} periodically analyzes the LLM task to model the relationship between cache size, load, performance, and power consumption (details in \Cref{subsec:profiler}). 
Second, a \emph{performance monitor} tracks the TTFT and TPOT of LLM serving and the current load (\ie request rate). 
Third, a combination of \emph{CI predictor} and \emph{load predictor} forecasts the CI and request rate based on historical values.
Specifically, we use the state-of-the-art CI predictor, EnsembleCI, \cite{yan2025ensembleci} for CI prediction (methodology details in \Cref{subsec:methodology}), and design a load predictor (details in \Cref{subsec:load_pred}).
A \emph{constraint solver} takes the SLO, predicted CI, performance, and power profile, predicted load, and current performance as input, and identifies the cache configuration that minimizes carbon emissions while attaining SLO (details in \Cref{subsec:ilp_solver}).
Finally, a \emph{cache manager} resizes the cache according to the cache configuration. It also incorporates a carbon-aware replacement policy that replaces cache entries with the Least Carbon Savings (\cachename{}) (details in \Cref{subsec:cache_controller}). 
We build \name{} on top of the LLM caching system, LMCache~\cite{lmcache}. 
\name{} can be flexibly adapted to different LLMs.

We next introduce the new components in the \name{} framework, marked green in \Cref{fig:overview}.

\subsection{Cache Performance Profiler}\label{subsec:profiler}

The cache performance profiler analyzes performance and power under different request rates and cache sizes for each LLM task.
The profiler samples a number of prompts and evaluates them on an initialized cache filled to maximum capacity.
The profiler consists of two components for evaluation, a monitoring tool and a carbon emission calculator, to profile each LLM task. 

\textbf{Performance and Power Analysis.}
The profiler sweeps cache sizes and request rates.
Cache size values are defined by the cache configurations, while request rates are varied up to the maximum level the system can support before violating SLOs.
For each combination of cache size and request rate, it records the TTFT and TPOT of every prompt. 
In parallel, it records the power consumption of key server components. 
The CPU power is measured using a RAPL tool \cite{ryzen-rapl} and GPU power is measured using pyNVML \cite{pynvml} every 1 ms.
For other relatively low-power components, like SSD and DRAM, that do not provide direct measurement interfaces, the monitoring tool follows the typical power in their specifications \cite{990pro,micronddr4}. 
In real-world LLM serving scenarios, profiling is performed periodically to adapt to dynamic changes in workload characteristics, such as distribution of request length and access patterns, as seen in prior work \cite{alpaserve}.

\textbf{Carbon Calculation.}
The carbon calculator follows the approach in \Cref{subsec:carbon-background,subsubsec:cache-carbon-modeling}. 
It first computes the total energy consumption using the timing and latency measurements and then derives the operational carbon based on the grid’s current carbon intensity, as defined in \Cref{eq:carbon-op}.
Next, it estimates the embodied carbon of the platform using \Cref{eq:carbon-eb}. Specifically, it calculates the embodied carbon of the SSD for caching via \Cref{eq:ssd-embodied}.
Finally, the total carbon emissions are obtained by summing the operational and embodied components, following \Cref{eq:cache+other}.

\begin{figure}
\centering
\begin{subfigure}[t]{0.48\linewidth}
    \centering
    \includegraphics[width=\linewidth]{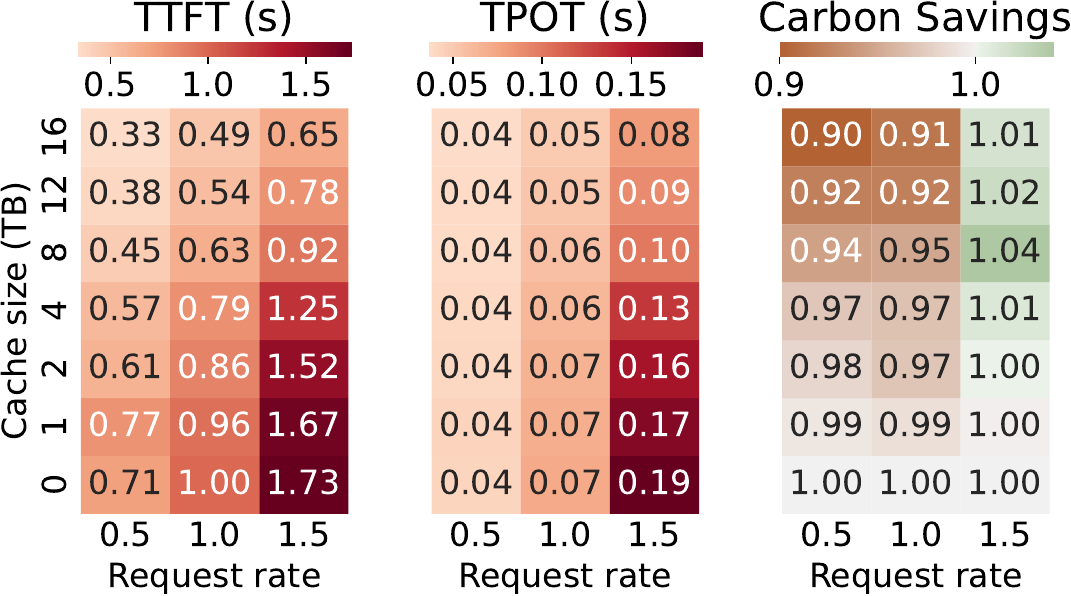}
    \caption{} \label{fig:70B_multiturn_heatmap} \Description{}
\end{subfigure}
\hfill
\begin{subfigure}[t]{0.48\linewidth}
    \centering
    \includegraphics[width=\linewidth]{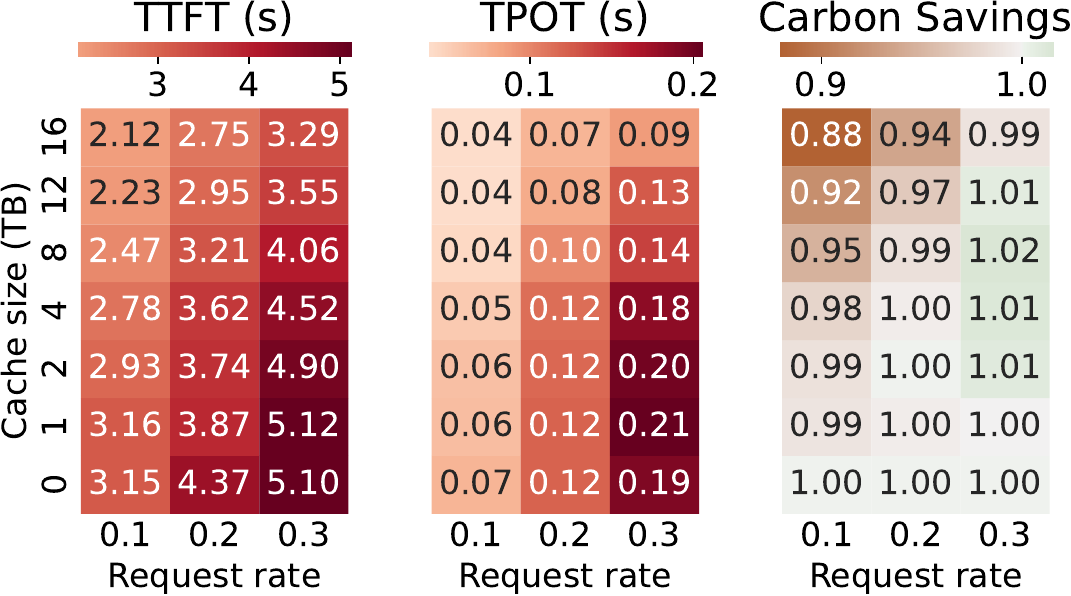}
    \caption{}  \label{fig:70B_QA_0.4_heatmap}  \Description{}
\end{subfigure}
\caption{Profiling results of TTFT and TPOT (lower is better), and carbon savings over no-cache (higher is better) when serving Llama-3 70B in the ES grid for (a) multi-turn conversation, using a ShareGPT dataset~\cite{sharegpt_dataset} and (b) document comprehension using a TriviaQA dataset \cite{joshi2017triviaqa} (with skewness of $\alpha=0.4$).} 
\Description{}
\end{figure}

In this work, we adapt two LLM tasks to \name{}, multi-turn conversation and document reading comprehension, which have also been studied in prior works \cite{cachegen,distserve,cachedattention,NEURIPS2024_LiuAdvances,chen2025broaden,Luo_2024_CVPR,Du_Nan_Zhang_Xie_Xu_Fan_Cui_Tao_Jiang_2024,jeongasplos25accelerating}.
For each task, we sample 500 prompts from the dataset and execute them after cache warmup. 
We demonstrate the profiling results of both tasks (ShareGPT dataset \cite{sharegpt} for multi-turn conversation and TriviaQA \cite{joshi2017triviaqa} for document reading comprehension) using a Llama-3 70B model \cite{llama3} on the platform described in \Cref{tab:embodied}.
As the profile targets \name{}, the cache uses our carbon-aware policy (details in \Cref{subsec:cache_controller}).
More experiment methodology details are described in \Cref{subsec:methodology}.

\Cref{fig:70B_multiturn_heatmap} shows heatmaps of the average TTFT and TPOT, and carbon savings from caching, under variable request rates (x-axis) and cache sizes (y-axis), for the multi-turn conversation task.
We profile TTFT and TPOT under various cache sizes (y-axis).
We also profile the carbon emission savings based on the CI of the ES grid --- defined as the ratio between the no-cache and cached configurations --- where a ratio greater than 1 indicates benefits.
The TTFT and TPOT in the profile are consistent with \Cref{box:rate-perf,box:size-perf}.
Overall, larger caches improve the latency, and the benefit is more prominent for higher request rates.  
The trend of carbon savings is also aligned with \Cref{box:rate-carbon}.
Larger caches may not bring carbon emission savings, especially when the request rate is low. 
When the rate is high, larger caches yield more savings.
\Cref{fig:70B_QA_0.4_heatmap} shows the heatmaps of the document reading comprehension task with skewness of $\alpha=0.4$ (\ie 10\,\% of documents are accessed by \textasciitilde 25\,\% of prompts).
The latency and carbon savings trends are consistent with the multi-turn conversation task. However, the longer document contexts result in a higher TTFT, which lowers the maximum request rate in this profile. 
These profiles will be used by a constraint solver (\Cref{subsec:ilp_solver}) to select the optimal cache size configuration.

\subsection{Load Predictor} \label{subsec:load_pred}

Prior work from Microsoft has shown that LLM prompts follow a similar pattern in a day \cite{stojkovic2025dynamollm}, where the rate mostly depends on the time of day.
We design a lightweight predictor, based on the Seasonal Autoregressive Integrated Moving Average (SARIMA) model, using \emph{pmdarima}~\cite{pmdarima}.
It captures both daily periodicity and short-term autocorrelation in the given load. 
The SARIMA parameters are estimated from historical rate traces using maximum likelihood fitting. We also use \emph{pmdarima} for parameter auto-tuning.
We use a hold-out evaluation, where it takes the most recent three consecutive days of data as input and predicts 24 hours ahead. 
Moreover, the load predictor performs online step-ahead prediction during runtime --- every hour, the model incorporates the most recent load before forecasting future hours to adapt to rate fluctuations. 
This predictor incurs minimum computation and runs on separate CPU cores, without interfering with the LLM workload.

\subsection{Constraint Solver}\label{subsec:ilp_solver}

\begin{table}[t]
\small
\centering
\setlength{\tabcolsep}{2.5pt}
\caption{Key notations in the ILP problem. }\label{tab:notation}
\begin{tabular}{ll ll}
\cmidrule[1pt](r{0.5em}){1-2} \cmidrule[1pt](l{0.5em}){3-4}
$t$ & Timestamp during the LLM execution & 
$p_{i}^{j_t}$  & Average power of request $i$ at request rate $j_t$\\ 
$\texttt{CI}_{t}$ & Carbon intensity (CI) at time $t$  & ${\rm TTFT}_{i}^{j_t}$  & TTFT for request $i$ at request rate $j_t$  \\
$j_{t}$ & Request rate at time $t$ & $S_{i}^{j_{t}}$   & Cache size of request $i$ at request rate $j_{t}$   \\
$N$ & Total number of requests in an LLM service & $\texttt{LT}_{\rm comp}$  & Lifetime of $\rm comp\in\{ {\rm SSD}, {\rm GPU}, {\rm CPU}, {\rm Mem} \}$ \\
\cmidrule[1pt](r{0.5em}){1-2} \cmidrule[1pt](l{0.5em}){3-4}
\end{tabular}

\end{table}

We first describe the formulation of the optimization problem. 
Then, we discuss our assumptions and potential error sources that can affect the decisions. 
Finally, we describe the solver's implementation.

\subsubsection{Problem Formulation}

\name{} minimizes total carbon emissions by selecting the optimal cache size while meeting the SLO attainment goal. 
The total carbon emissions include both operational carbon and embodied carbon emissions, where embodied emissions include storage (SSD) and non-storage (GPU, CPU, and memory). 

The carbon optimization focuses on the prefill phase, where cache hits directly reduce computation by eliminating redundant processing for existing contexts.
While cache size does not directly reduce auto-regressive token generation computations in the decode phase, it effectively reduces decode latency by accelerating the prefill phase, mitigating decode delays for the concurrently scheduled requests within a continuous batching system.
For this reason, as discussed in \Cref{subsec:caching-background}, caching influences only the waiting time rather than the total computation of the decode phase. 

We formulate this as an Integer Linear Programming (ILP) problem. The key notations are summarized in Table~\ref{tab:notation}.
An LLM service handles a total of $N$ requests.
At each time $t$, the LLM serving system experiences a request rate of $j_t$. 
Each request $i$ has power consumption $p_{i}^{j_t}$, time-to-first-token latency ${\rm TTFT}_{i}^{j_t}$, where the request rate of a future time $t$ is predicted using our load predictor. 
Given the predicted carbon intensity $\texttt{CI}_t$ of a future time $t$, the operational carbon emissions for the prefill phase can be calculated by multiplying power, TTFT, and CI. 
Because of adaptive caching, the calculation of SSD embodied carbon emissions is different from other hardware components, as shown in \Cref{eq:ssd-embodied}.
It depends on both the cache size $S_{i}^{j_{t}}$ and execution time $\rm TTFT_{i}^{j_{t}}$.
We calculate the embodied carbon of the SSD capacity allocated to the LLM service.

Then, given the execution time, we calculate the fraction of time a prompt utilizes the SSD's lifetime.
\Cref{eq:ilp} summarizes the objective function and constraints for the ILP problem. For an LLM service with $N$ requests in total with request rate $j_t$ at time $t$, \name{} minimizes total carbon emissions while meeting TTFT and TPOT SLOs as follows:
\begin{align}\label{eq:ilp}
\argmin_{S_{i}^{j_{t}}} & \sum_{i=1}^N \bigg(  \overbrace{p_{i}^{j_t} {\rm TTFT}_{i}^{j_t}\texttt{CI}_t}^\textbf{Operational carbon}  + \overbrace{\frac{{\rm TTFT}_{i}^{j_t}}{\texttt{LT}_{SSD}} S_{i}^{j_{t}} C_{\rm e, SSD}^{\rm Unit}}^\textbf{Cache embodied carbon} +  \overbrace{\sum_{{\rm comp}\in{{\rm GPU}, {\rm CPU}, {\rm Mem}}}\frac{{\rm TTFT}_{i}^{j_t}}{\texttt{LT}_{\rm comp}}C_{\rm e, comp}}^\textbf{Other embodied carbon} \bigg) \nonumber \\
s.t. \quad & \sum_{i=1}^N z_{\rm TTFT, i} \geq \rho N \bigwedge \sum_{i=1}^N z_{\rm TPOT, i} \geq \rho N~,
\end{align}
where $z_{\rm TTFT, i}, z_{\rm TPOT, i}\!\!\in\!\!\{0,1\}$ are binary variables such that $z_{TTFT, i}\!\!=\!\!1$ if a request $i$ meets the TTFT constraint, and $z_{\rm TPOT, i}\!\!=\!\!1$ means request $i$ meets the TPOT constraint.
$\rho$ specifies the required fraction of prompts that meet the requirements.
We set it as 0.9, meaning at least 90\,\% of requests must satisfy both TTFT and TPOT latency requirements, corresponding to the targeted SLO attainment.

\subsubsection{Assumptions and Error Analysis}

To formulate the optimization problem efficiently, we make three key assumptions. 
First, we constrain the cache size variable $S_t$ to a discrete set of integers, which aligns with the granularity of cache sizes in the cloud. This constraint disallows the selection of fine-grained intermediate values.
Second, we assume that the carbon intensity ($\texttt{CI}_t$) remains constant within each decision interval (1 hour due to the granularity of the CI dataset \cite{maji2022carboncast}).
Third, we assume that the profiled metrics (e.g., TTFT) reflect the system's performance under stable conditions.
In particular, we collect the offline profiling data after a cache warm-up period using our LCS replacement policy. Therefore, the dependencies between cache sizes, cache hit rate, and the resultant latency are implicitly captured within the profiled results.
Satisfying the SLO constraint requires exploring a combinatorial search space, scaling as $O(2^T)$ in the worst case, where $T$ denotes the count of discrete optimization points defined by the granularity of timestamps. 
We further prove that this optimization problem is NP-hard via a reduction from the \textsc{0--1 Knapsack} problem in \Cref{sec:ilp_comp}.
However, the specific problem in \name{} is computationally tractable due to the constrained decision space. Unlike the general Knapsack problem, our formulation involves a limited set of discrete cache size candidates (1 TB granularity) and a finite time interval. 
These practical assumptions effectively prune the search space. 
We accept the resulting ``rounding loss'' as a necessary tradeoff to maintain the computational efficiency required for online decision-making.

In addition, we acknowledge three types of errors affecting theoretical optimality.
First, the load predictor may introduce errors compared to the ideal outcomes. 
Second, the CI predictor can also introduce errors.
Third, profiling error arises from discrepancies between the profiling phase and actual execution. 
We evaluate these errors in \Cref{subsec:error_analysis} and demonstrate that prediction and profiling errors have a low impact on total carbon savings.

\subsubsection{Solver Implementation}

We use the PuLP optimization modeling library \cite{ilpsolver} with the COIN-OR CBC solver \cite{cbc} as the backend to solve the ILP objective function of \Cref{eq:ilp}.
\name{} takes the optimal cache configuration determined by the ILP solver and performs cache resizing every hour.
The cache has an allocation granularity of 1 TB, with a maximum size of 16 TB. 
We evaluate the ILP execution time in \Cref{subsec:ilp_overhead} --- 7.03~s per decision on average, a low overhead compared to the hourly cache resizing frequency.

\subsection{Cache Manager and Replacement Policy}
\label{subsec:cache_controller}

The cache controller takes the cache size configuration from the constraint solver and manages the cache. 
To enlarge the cache, the cache controller allocates more SSD space from the cloud. 
When the cache shrinks, the cache controller evicts existing cache entries with the lowest scores until the total size reaches the configuration. Then, the spare cache space will be released.

The replacement policy determines the score of cache entries.
To optimize carbon efficiency, \name{} introduces \textit{Least Carbon Savings} (\cachename{}), a new carbon-aware cache replacement policy that prioritizes entries based on their potential carbon impact. 
Unlike conventional cache replacement policies that focus on either access recency (e.g., LRU) or frequency (e.g., LFU), \cachename{} minimizes carbon emissions by evicting cache entries that offer the least carbon savings.
As discussed in~\Cref{subsec:ilp_solver}, carbon emissions from caching depend on both cache size and hit rate. We design a scoring function guided by four key insights:
\begin{enumerate}[leftmargin=*,label={(\roman{enumi})}]
    \item Prioritize entries with more hit tokens, which yield higher operational carbon savings, as shown in \Cref{box:length-perf}. 
    \item Favor frequently accessed entries, as they are more likely to produce cache hits and contribute to greater carbon savings.
    \item Prefer smaller entries that consume less storage, reducing embodied carbon from SSD usage. Although a longer context reduces more operational carbon as concluded in \Cref{box:length-perf}, it costs more embodied carbon, which introduces a tradeoff. 
    \item Consider recency of access, as recently used entries are more likely to be accessed again soon. 
\end{enumerate}
Together, we design the \cachename{} score as follows
\begin{align}\label{eq:overall-score}
\mathit{Score} = \frac{ \#\mathit{Token} \times \#\mathit{Hit}}{\mathit{Size}\times \mathit{Age}}~.
\end{align}
$\#\mathit{Token}$ is the accumulated hit token number of this cache, representing the volume of reused context. A higher number indicates more operational carbon savings, supporting \textit{Insight~(i)}. $\#\mathit{Hit}$ is the number of cache hits. Frequent access increases the chance of carbon savings by avoiding recomputation, supporting \textit{Insight~(ii)}. $\mathit{Size}$ is the cache size of the entry. Dividing by cache size encourages keeping smaller entries that consume less embodied carbon, supporting \textit{Insight~(iii)}. $\mathit{Age}$ is how long the cache has stayed in storage, where older entries are less likely to be reused. Penalizing age promotes the eviction of stale entries, supporting \textit{Insight~(iv)}.

As described in \Cref{subsec:profiler}, we evaluate two LLM tasks. 
Overall, $\mathit{Size}$ and $\mathit{Age}$ in \Cref{eq:overall-score} work in the same way for both tasks, but we specifically adapt $\#\mathit{Token}$ and $\#\mathit{Hit}$ fields to each task.

\noindent\textbf{Task 1: Multi-turn conversation.} 
Each turn reuses the KV cache of context from previous turns. 
We adapt \Cref{eq:overall-score} to prioritize cache entries that contribute more to carbon savings: 
\begin{align}\label{eq:multiturn}
\mathit{Score} = \frac{\mathit{CurTurn} \times \#\mathit{AccuToken}}{\mathit{Size} \times \mathit{Age}}~,  
\end{align}
where $\mathit{CurTurn}$ encourages retaining cache entries deeper in the conversation, since later turns rely more heavily on prior context. $\#\mathit{AccuToken}$ reflects the total reused tokens across turns, directly tied to operational carbon savings. 

\noindent\textbf{Task 2: Document reading comprehension.}
A document is reused across multiple questions, similar to turns in a dialogue.
We treat these questions as equivalent to conversation turns and adapt the score accordingly:
\begin{align}\label{eq:document}
\mathit{Score}=\frac{\#\mathit{Hit} \times \mathit{AccuDocLen}}{\mathit{Size} \times \mathit{Age}}~,
\end{align}
where $\#\mathit{Hit}$ tracks the number of times a document is reused across questions, and $\mathit{AccuDocLen}$ measures the total reused document length, capturing the operational carbon savings. 
\section{Evaluation}

\subsection{Methodology} \label{subsec:methodology}

\textbf{Experiment setup.}
We evaluate two models, Llama-3 70B and 8B \cite{llama3}. The 70B model runs on $4\times$ NVIDIA L40 GPUs using the platform in \Cref{tab:embodied}, and the 8B model reduces the GPUs to $2\times$ L40 as the model is less demanding. 
Due to GPU memory constraints, the 70B model runs under INT8. The 8B model runs under the default BF16.
Both models have a context window of 8k tokens. When the context goes beyond this limit, we truncate extra context like prior work~\cite{cachedattention,cachegen}.
The platform maintains a maximum capacity of 16 TB of SSD. \name{} provisions SSD at a granularity of 1 TB to the LLM caching system. 
The power measurement method follows \Cref{subsec:profiler}. 
The caching system is built on top of LMCache \cite{lmcache}, which includes vLLM~\cite{pagedattention} and continuous batching~\cite{orca} optimizations. 
We also integrate our carbon-aware \cachename{} replacement policy into LMCache.  
We use a maximum cache size of 16 TB for the 70B model and 8 TB for the 8B model. 
By default, \name{} resizes the cache every hour. \Cref{subsubsec:switch} evaluates different resizing frequencies. 

\textbf{Tasks and Datasets.}
We evaluate two datasets that correspond to the two LLM tasks commonly used by prior works~\cite{cachegen,distserve,cachedattention,NEURIPS2024_LiuAdvances,chen2025broaden,Luo_2024_CVPR,Du_Nan_Zhang_Xie_Xu_Fan_Cui_Tao_Jiang_2024,jeongasplos25accelerating}. In both datasets, we use distinct sets of prompts for profiling and evaluation.

\begin{itemize}[leftmargin=*]
    \item \textbf{Multi-turn conversation based on ShareGPT \cite{sharegpt,sharegpt_dataset}.} 
    We randomly select a conversation every time and take its next conversation turn as the input prompt. 
    The request follows a Poisson distribution like prior works~\cite{cachedattention,hcache,distserve}.  
    By varying \textalpha{}, the request generator simulates different average request rates. 
    We initialize the cache with 200k prompts. 
    \item \textbf{Document comprehension based on TriviaQA \cite{joshi2017triviaqa}.}
    Because TriviaQA is initially used for training, it has an almost uniform number of visits per document. Therefore, we introduce skewness by following Zipf distributions similar to prior studies on data caching \cite{berg2020cachelib,gill2007youtube,hasan2014tradeoffscdn,qi2013facebookphoto,breslau1999webcaching}. 
    We evaluate Zipf with low and high skewness levels:  \textalpha{}=0.4 (10\,\% of documents are accessed by \textasciitilde 25\,\% of prompts) and \textalpha{}=0.7 (10\,\% of documents are accessed by \textasciitilde 50\,\% of prompts).
    Prompts are also generated under Poisson distributions.
    As the context is larger than multi-turn conversations, we initialize the cache with 50k prompts.
\end{itemize}

\textbf{Request rate.}
Prior work has shown that the request rate to LLM services varies in a day \cite{patel2024splitwise,stojkovic2025dynamollm}. To reflect realistic load dynamics, we use the Azure LLM trace \cite{azurellm2024} to simulate prompt arrivals at corresponding request rates. 
We downscale the request rate of the Azure trace to match our platform’s capacity, ensuring that the peak rate is still within our system's sustainable throughput.
The load predictor in \Cref{subsec:load_pred} forecasts future request rates using past and current rate data. 

\textbf{Carbon intensity.}
We evaluate four grids, FR, FI, ES, and CISO, and predict their CI using EnsembleCI \cite{yan2025ensembleci}, which is trained with 18 months of data and predicts CI for July 6, 2022.
The training data and groundtruth CI come from the CarbonCast dataset \cite{maji2022carboncast}. 
We use the CI on this date for the main performance and carbon evaluation, where we incorporate the variable CI traces with the variable request rate to simulate a realistic scenario. 
We also use the average CI of the ES grid for sensitivity and ablation studies.

\textbf{SLOs.}
We define SLOs for both TTFT and TPOT. 
For multi-turn conversations, we set TTFT and TPOT SLOs to 2.5~s and 0.2 s for the 70B model, and 0.5 s and 0.15 s for the 8B model. 
For document comprehension, where inputs are longer and latency is less time-critical, we relax the TTFT thresholds to 15 s for the 70B model and 2.5 s for the 8B model. 
These SLOs are aligned with prior work~\cite{distserve,agrawl2024sarathiserve} and benchmarking results from LLMPerf~\cite{llmperf}.
We focus on an SLO attainment of 90\,\%, requiring at least 90\,\% requests to meet the timing constraints.

\textbf{Comparison points.}
We evaluate the following system design points: 
    (1) \textbf{No Cache}: Non-caching baseline with vLLM and continuous batching.
    (2) \textbf{Full Cache}: Use the maximum cache sizes as described in the experiment setup.
    (3) \textbf{\name{}}: The carbon-aware caching system in this work. The maximum cache size is the same as Full Cache.

\subsection{Carbon Emission and SLO Attainment} \label{subsec:carbon_slo}

\begin{figure}[t]
\begin{subfigure}[b]{1\linewidth}
    \centering
    \begin{subfigure}[b]{1\linewidth}
    \centering
    \includegraphics[width=1\linewidth]{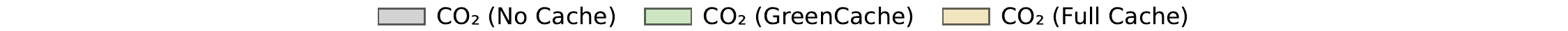}
    \end{subfigure}
    \begin{subfigure}[b]{0.32\linewidth}
    \centering
    \includegraphics[width=1\linewidth]{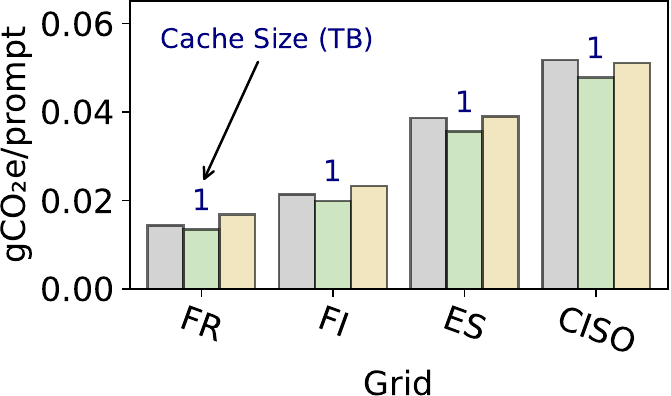}
    \caption*{Multi-turn conversation.}
    \end{subfigure}
    \hspace{1mm}
    \begin{subfigure}[b]{0.32\linewidth}
    \includegraphics[width=1\linewidth]{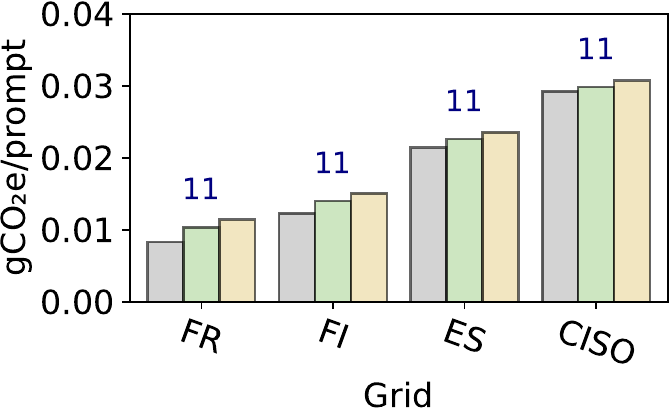}
    \caption*{Document comp. ($\alpha=0.4$).}
    \end{subfigure}
    \hspace{1mm}
    \begin{subfigure}[b]{0.32\linewidth}
    \includegraphics[width=1\linewidth]{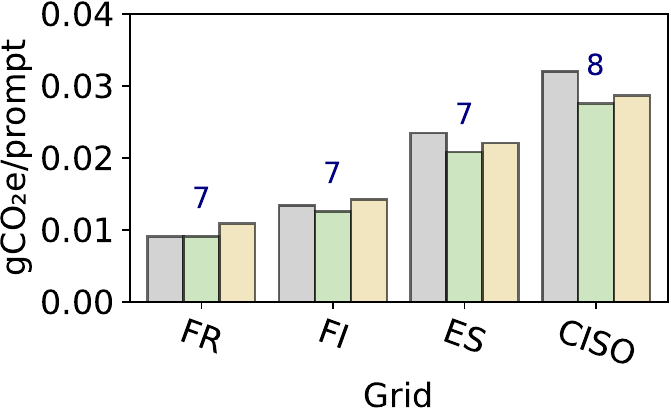}
    \caption*{Document comp. ($\alpha=0.7$).}
    \end{subfigure}
    \caption{Carbon emissions of Llama-3 70B.}
    \label{fig:70B_aggregated}
    \Description{}
\end{subfigure}
\begin{subfigure}[t]{1\linewidth}
    \centering
    \begin{subfigure}[b]{0.32\linewidth}
    \centering
    \includegraphics[width=1\linewidth]{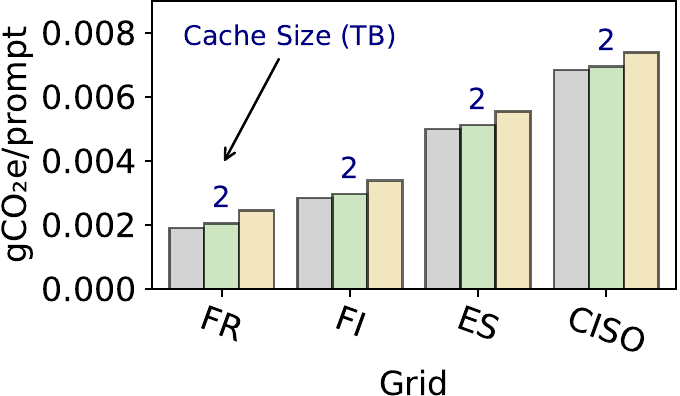}
    \caption*{Multi-turn conversation.}
    \end{subfigure}
    \hspace{1mm}
    \begin{subfigure}[b]{0.32\linewidth}
    \includegraphics[width=1\linewidth]{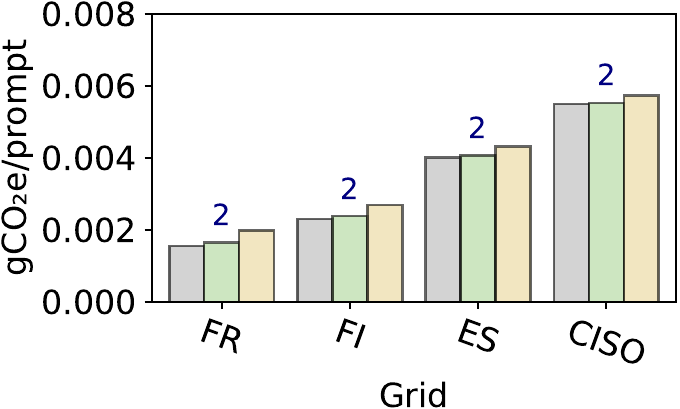}
    \caption*{Document comp. ($\alpha=0.4$).}
    \end{subfigure}
    \hspace{1mm}
    \begin{subfigure}[b]{0.32\linewidth}
    \includegraphics[width=1\linewidth]{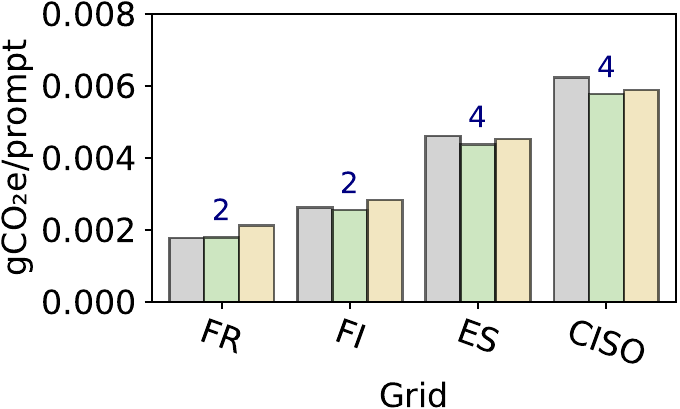}
    \caption*{Document comp. ($\alpha=0.7$).}
    \end{subfigure}
    \caption{Carbon emissions of Llama-3 8B.}
    \label{fig:8B_aggregated}
    \Description{}
\end{subfigure}
\caption{Average carbon emissions of LLM tasks. }
\end{figure}

\begin{figure}
\begin{subfigure}[t]{1\linewidth}
    \centering
    \begin{subfigure}[b]{1\linewidth}
    \centering
    \includegraphics[width=1\linewidth]{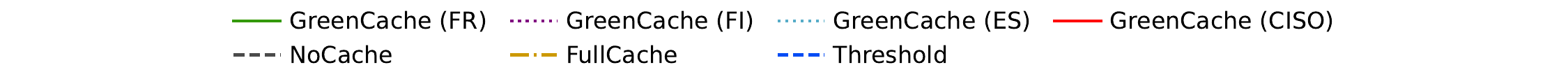}
    \end{subfigure}
    \begin{subfigure}[b]{0.32\linewidth}
    \centering
    \includegraphics[width=1\linewidth]{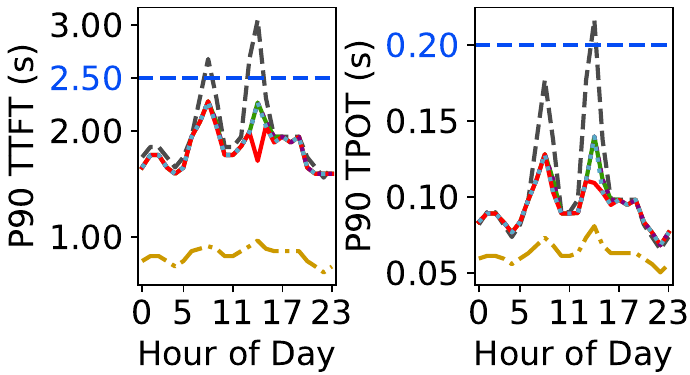}
    \caption*{Multi-turn conversation.}
    \end{subfigure}
    \hfill
    \begin{subfigure}[b]{0.32\linewidth}
    \includegraphics[width=1\linewidth]{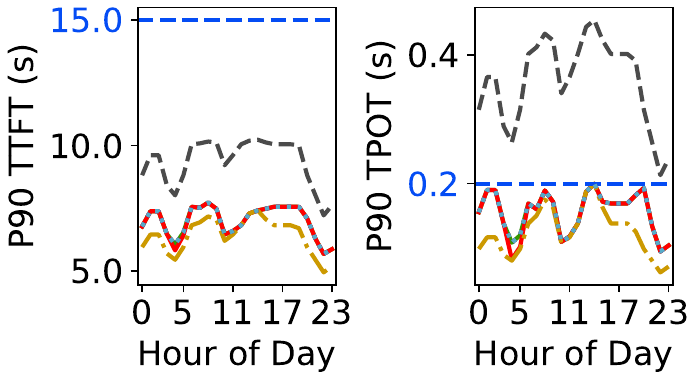}
    \caption*{Document comp. ($\alpha=0.4$).}
    \end{subfigure}
    \hfill
    \begin{subfigure}[b]{0.32\linewidth}
    \includegraphics[width=1\linewidth]{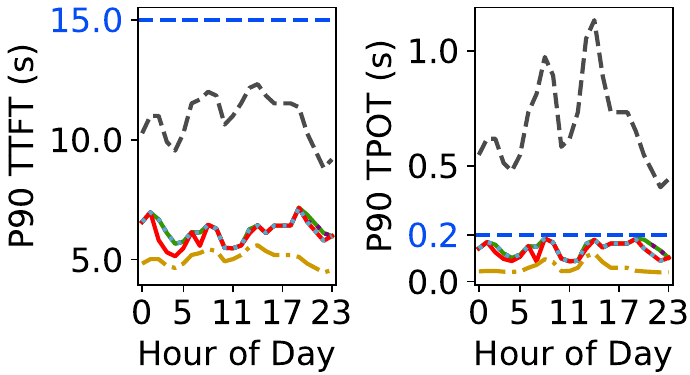}
    \caption*{Document comp. ($\alpha=0.7$).}
    \end{subfigure}
    \caption{Llama-3 70B.}
    \label{fig:70B_SLO}
    \Description{}
\end{subfigure}
\begin{subfigure}[t]{1\linewidth}
    \centering
    \begin{subfigure}[b]{0.32\linewidth}
    \centering
    \includegraphics[width=1\linewidth]{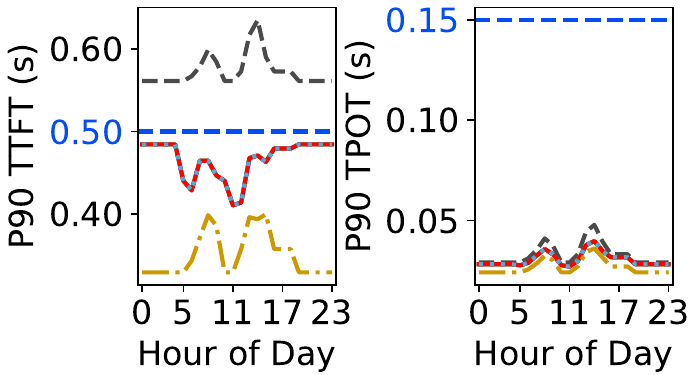}
    \caption*{Multi-turn conversation.}
    \end{subfigure}
    \hfill
    \begin{subfigure}[b]{0.32\linewidth}
    \includegraphics[width=1\linewidth]{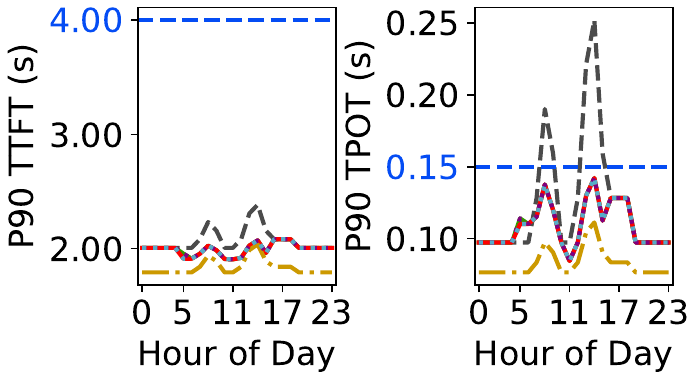}
    \caption*{Document comp. ($\alpha=0.4$).}
    \end{subfigure}
    \hfill
    \begin{subfigure}[b]{0.32\linewidth}
    \includegraphics[width=1\linewidth]{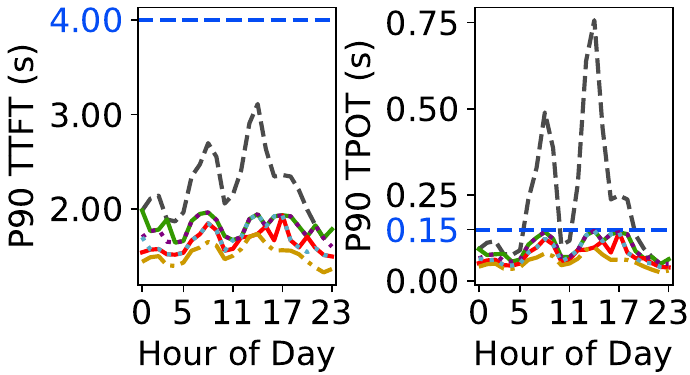}
    \caption*{Document comp. ($\alpha=0.7$).}
    \end{subfigure}
    \caption{Llama-3 8B.}
    \label{fig:8B_SLO}
    \Description{}
\end{subfigure}
\caption{SLO attainment timelines of LLM tasks. 
}
\end{figure}

We evaluate the carbon emissions and SLO attainment of \name{} using dynamic CI traces of four grids, where each grid follows an Azure request rate trace. 

\begin{figure}
\begin{subfigure}[t]{1\linewidth}
    \centering
    \begin{subfigure}[b]{1\linewidth}
    \centering
    \includegraphics[width=0.75\linewidth]{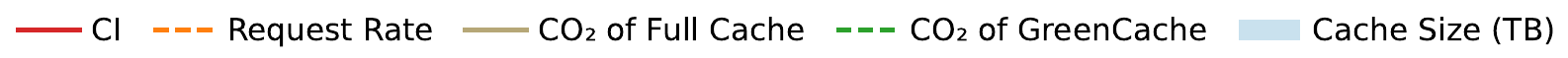}
    \end{subfigure}
    \begin{subfigure}[b]{0.48\linewidth}
    \includegraphics[width=1\linewidth]{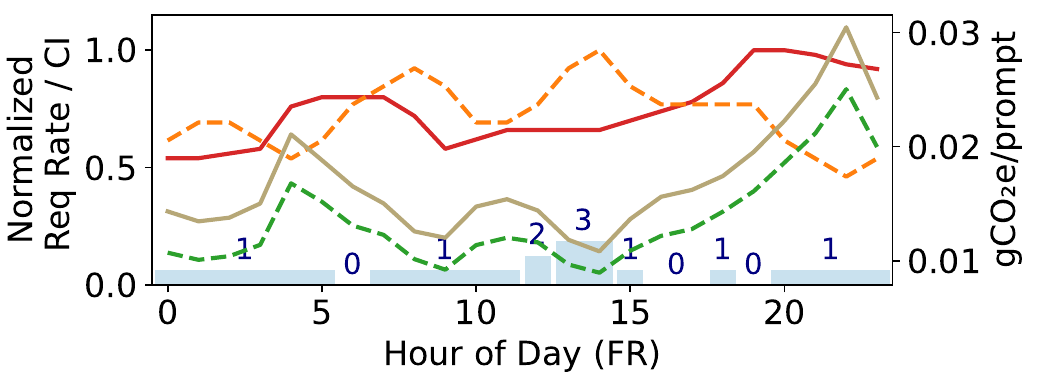}
    \end{subfigure}
    \hfill
    \begin{subfigure}[b]{0.48\linewidth}
    \includegraphics[width=1\linewidth]{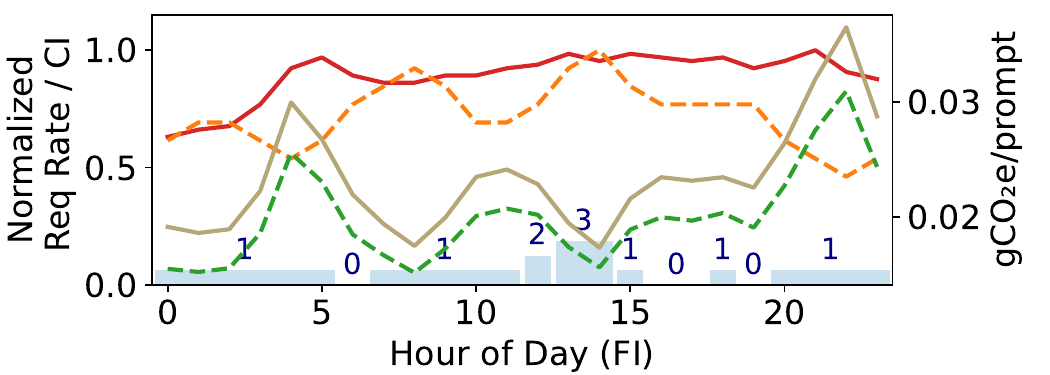}
    \end{subfigure}
    \begin{subfigure}[b]{0.48\linewidth}
    \includegraphics[width=1\linewidth]{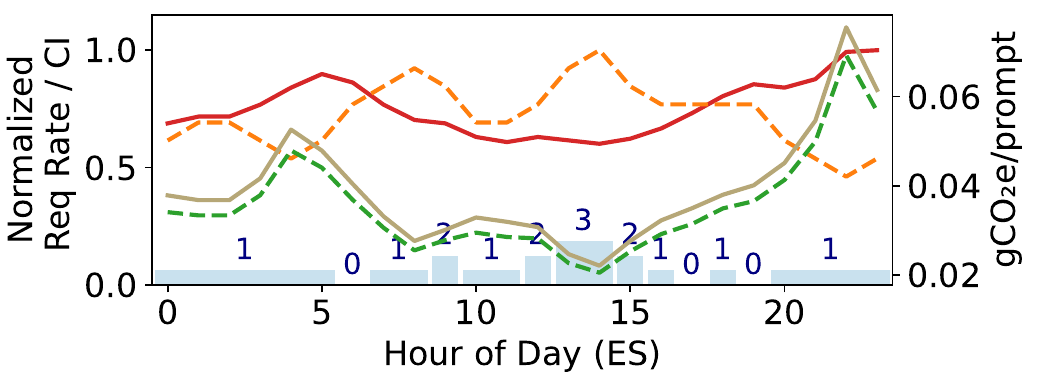}
    \end{subfigure}
    \hfill
    \begin{subfigure}[b]{0.48\linewidth}
    \includegraphics[width=1\linewidth]{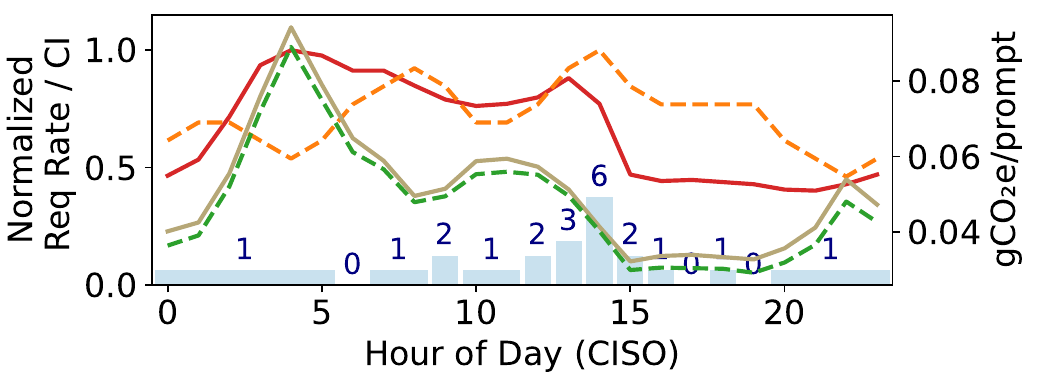}
    \end{subfigure}
    \caption{Multi-turn conversation with Llama-3 70B.}
    \label{fig:multiturn_70B}
    \Description{}
\end{subfigure}
\begin{subfigure}[t]{1\linewidth}
    \centering
    \begin{subfigure}[b]{0.48\linewidth}
    \includegraphics[width=1\linewidth]{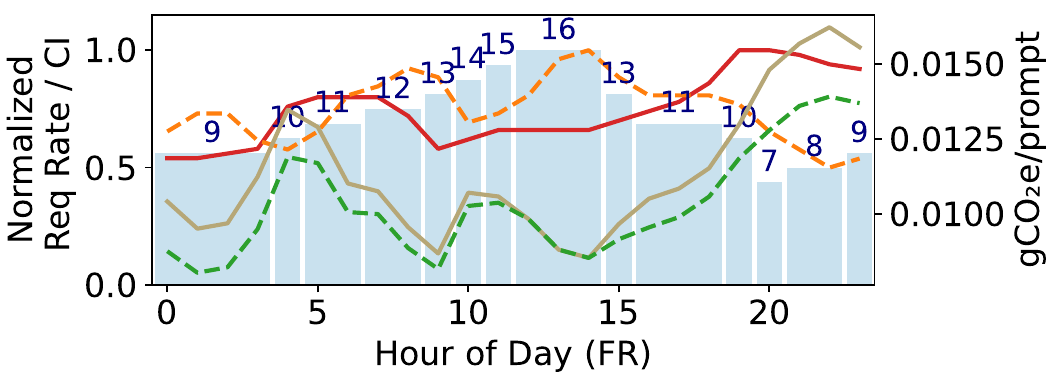}
    \end{subfigure}
    \hfill
    \begin{subfigure}[b]{0.48\linewidth}
    \includegraphics[width=1\linewidth]{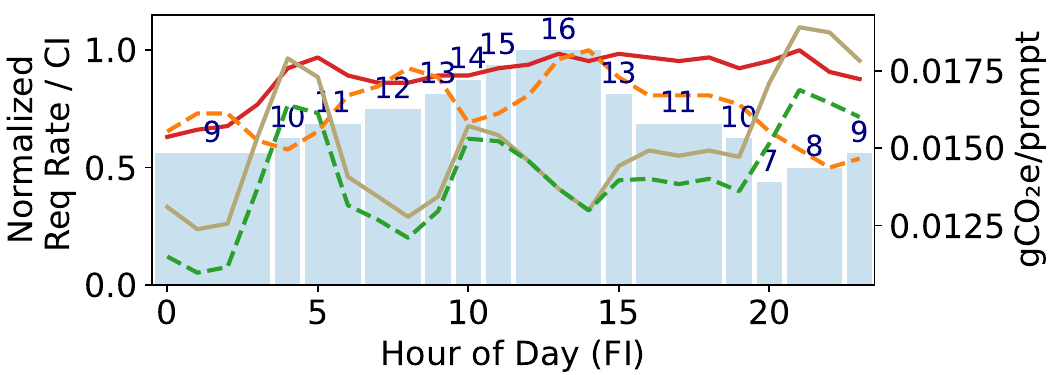}
    \end{subfigure}
    \begin{subfigure}[b]{0.48\linewidth}
    \includegraphics[width=1\linewidth]{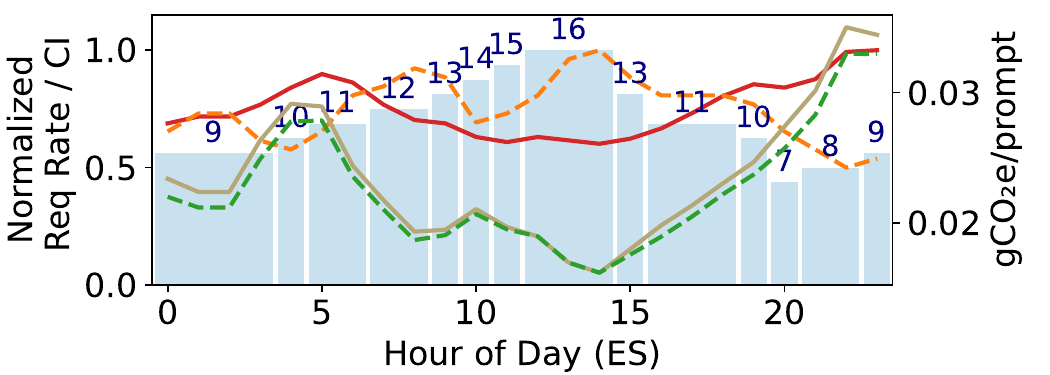}
    \end{subfigure}
    \hfill
    \begin{subfigure}[b]{0.48\linewidth}
    \includegraphics[width=1\linewidth]{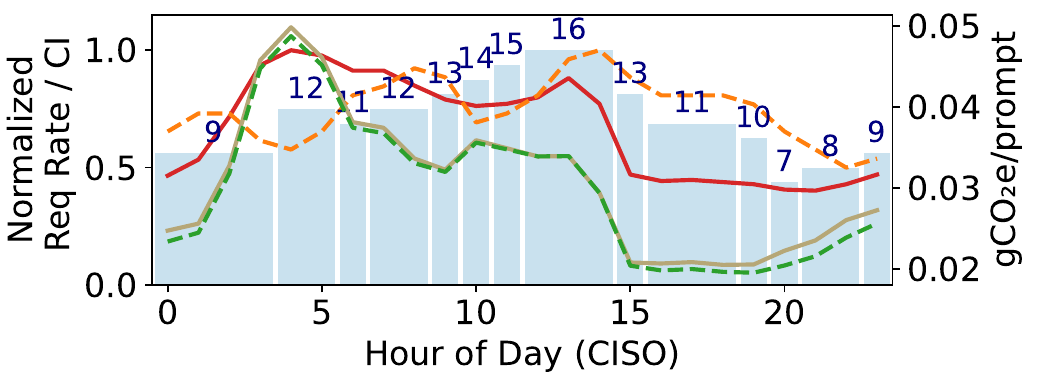}
    \end{subfigure}
    \caption{Document comprehension ($\alpha=0.4$) with Llama-3 70B.}
    \label{fig:qa_04_70B}
    \Description{}
\end{subfigure}
\caption{Timelines of carbon emissions under variable CI and rate. }
\end{figure}

\textbf{Carbon Emissions.}
First, we present the overall carbon emissions when serving Llama-3 70B and 8B models, as \Cref{fig:70B_aggregated,fig:8B_aggregated} show, respectively. 
Compared to Full Cache, \name{} achieves average 12.6\,\%, 9.4\,\%, and 5.6\,\% lower carbon in multi-turn conversation and document comprehension with two skewness levels, respectively, when serving Llama-3 70B. 
In comparison, when serving Llama-3 8B, the carbon emission reductions are slightly lower -- 10.8\,\%, 7.6\,\%, and 9.7\,\%, respectively, as the 8B model is lightweight and requires a smaller cache. 
The carbon emission savings come from a reduced cache usage over Full Cache, as the cache size labels indicate. 
We also observe that \name{} achieves higher savings in low-CI grids like FR and FI, 14.4--20.3\,\%, as embodied carbon is more prominent in these grids. 
The skewness of the document comprehension dataset also significantly changes the carbon savings and cache sizes. 
The high-skewness case ($\alpha=0.7$) requires a smaller cache size compared to the low-skewness case ($\alpha=0.4$) as fewer prompts are frequently used. 
Thus, the carbon emission savings from \name{} in the high-skewness case are higher. 
We notice that No Cache sometimes achieves lower carbon emissions than \name{}. However, this option is not acceptable due to the violation of SLOs, which will be discussed next.

\textbf{SLO Attainment.}
We evaluate \name{}'s SLO attainment against baselines in the four grids.
\Cref{fig:70B_SLO,fig:8B_SLO} present the P90 TTFT and TPOT through a day of both models, and compare them against the thresholds as specified by the SLOs. 
The P90 latency staying below the SLO-specified thresholds indicates at least 90\,\% SLO attainment. 
Among all scenarios, \name{} only exhibits slightly higher P90 latency than Full Cache, staying below both TTFT and TPOT thresholds as specified by the SLOs, indicating over 90\,\% SLO attainment.
In contrast, No Cache's P90 latency exceeds the SLO constraints frequently, which is not considered a viable solution.

\textbf{Timelines.}
Finally, we showcase timelines that demonstrate the dynamics of cache size and per-prompt carbon emissions under real-time CI and request rate for Llama-3 70B. 
\Cref{fig:multiturn_70B,fig:qa_04_70B} show the timelines for the multi-turn conversation and document comprehension (with skewness $\alpha=0.4$) tasks, respectively. 
The y-axis presents the CI and request rate normalized by their highest values.  
We do not include the No Cache baseline in the timeline as it fails to meet SLO.
Throughout the day, \name{} reduces carbon emissions by 6.9--20.4\,\% and 3.2--16.2\,\% over Full Cache in multi-turn conversation and document comprehension, respectively. 
Among the 4 grids, FR demonstrates the most carbon reduction due to its low CI, amplifying the impact of reducing embodied carbon; \name{} reduces carbon emissions by an average of 15.1\,\% and up to 25.3\,\% reduction in multi-turn conversation, as \name{} utilizes 15 TB less SSD than Full Cache (12 AM of the day in \Cref{fig:multiturn_70B}).
Specifically, cache size is more sensitive to the request rate under low CI, as a larger cache is needed to meet the SLO attainment goal. 
In contrast, under high CI, the cache sizes are generally larger, as caching is effective at reducing operational carbon. 
For example, in \Cref{fig:multiturn_70B}, FR, FI, and ES grids have similarly small cache sizes, while CISO has more variation in cache sizes. 
In \Cref{fig:qa_04_70B}, the difference is less prominent as this task has longer contexts and requires larger caches. 
For the same reason, \name{} saves less carbon when the load is higher because a larger cache is needed (up to the full 16 TB of cache).

\subsection{Ablation Studies}

In this section, we first study benefits from the adaptive caching and then compare \cachename{} with other cache replacement policies. 

\begin{figure}

\begin{subfigure}[b]{1\linewidth}
    \centering
    \includegraphics[width=0.3\linewidth]{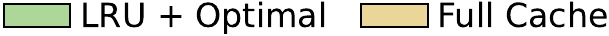}
\end{subfigure}

\begin{subfigure}[b]{0.28\linewidth}
    \centering
    \includegraphics[width=\linewidth,trim={0 0 0 16mm},clip]{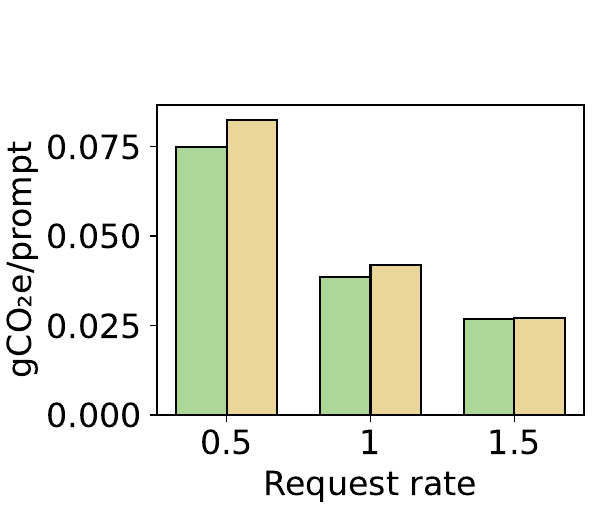}
    \caption{Multi-turn conversation.}
    \label{fig:ablation_LRU_70B_multiturn}
    \Description{}
\end{subfigure}
\hfill
\begin{subfigure}[b]{0.28\linewidth}
    \centering
    \includegraphics[width=\linewidth,trim={0 0 0 16mm},clip]
    {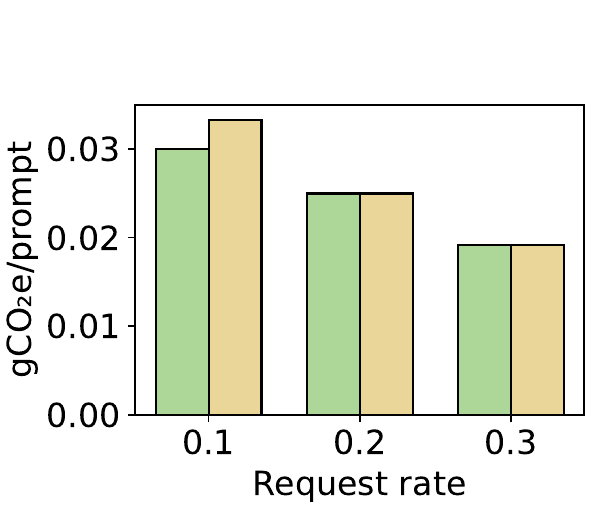}
    \caption{Document comp. ($\alpha$=0.4).}
    \label{fig:ablation_LRU_70B_QA0.4}
    \Description{}
\end{subfigure}
\hfill
\begin{subfigure}[b]{0.28\linewidth}
    \centering
    \includegraphics[width=\linewidth,trim={0 0 0 16mm},clip]
    {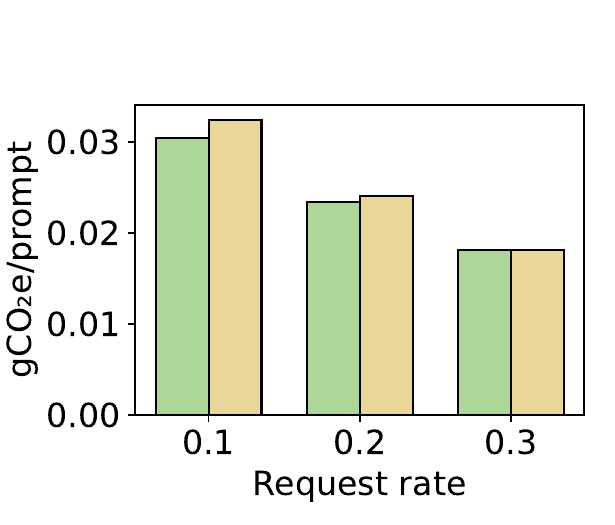}
    \caption{Document comp. ($\alpha$=0.7).}
    \label{fig:ablation_LRU_70B_QA0.7}
    \Description{}
\end{subfigure}
\caption{Ablation study on adaptive caching (Llama-3 70B). }
\label{fig:ablation_study}
\Description{}
\end{figure}

\subsubsection{Adaptive Caching Analysis}
We first conduct an experiment that integrates \name{}'s adaptive caching technique into LMCache \cite{lmcache} while using its original LRU policy (\ie LRU + Optimal). \Cref{fig:ablation_study} shows the carbon reduction under different request rates. 
For this analysis, we use the average carbon intensity of the ES grid, which is 124 \unitCI{}. 
Compared to Full Cache, \name{}'s adaptive caching can bring up to 10.3\,\% carbon savings in multi-turn conversation and 6.6--9.9\,\% reduction in document comprehension.
As the request rate increases, the benefit decreases because \name{} adapts to larger cache sizes to meet the SLO attainment goal. 

\subsubsection{Replacement Policy Comparison}
We then compare the hit rates of the following replacement policies:
(1) \textbf{First In, First Out (FIFO)} replaces blocks in the order in which they were added.
(2) \textbf{Least Recently Used (LRU)} replaces the least recently used cache block. It is the default policy in LMCache \cite{lmcache}.
(3) \textbf{Least Carbon Savings (\cachename{})} is the replacement policy in this work that replaces the entry that brings the least carbon emission savings (details in \Cref{subsec:cache_controller}). 

We evaluate cache sizes ranging from 1 to 16 TB with the Llama-3 70B model.
We define \emph{cache hit rate} as the number of tokens reused from the cache over the total number of input tokens, as shown in \Cref{tab:replacement}.
In the vast majority of cases, our replacement policy \cachename{} outperforms other replacement policies. Although LRU indicates similar performance as \cachename{} with 16TB cache size, \cachename{} demonstrates up to 9\,\% higher hit rate than LRU with smaller cache sizes.
We also notice that, in the document comprehension task, higher skewness ($\alpha=0.7$) leads to better cache hit rates in all policies, compared to low skewness ($\alpha=0.4$).

\begin{figure}
\begin{minipage}[c]{0.64\linewidth}
    \centering  
    \small
    \captionof{table}{Hit rate (Llama-3 70B). }
    \label{tab:replacement}
    \setlength{\tabcolsep}{3pt}
        \begin{tabular}{c ccc ccc ccc}
    \toprule
    \multirow{2}{*}{\makecell{Cache\\Size\\(TB)}} & \multicolumn{3}{c}{ShareGPT}  & \multicolumn{3}{c}{\makecell{TriviaQA,$\alpha$=0.4}}  & \multicolumn{3}{c}{\makecell{TriviaQA,$\alpha$=0.7}} \\
         \cmidrule(lr){2-4} \cmidrule(lr){5-7}  \cmidrule(lr){8-10} 
        & \rotatebox[origin=c]{90}{FIFO} & \rotatebox[origin=c]{90}{LRU} & \rotatebox[origin=c]{90}{\cachename{}}  & \rotatebox[origin=c]{90}{FIFO} & \rotatebox[origin=c]{90}{LRU} & \rotatebox[origin=c]{90}{\cachename{}}  & \rotatebox[origin=c]{90}{FIFO} & \rotatebox[origin=c]{90}{LRU} & \rotatebox[origin=c]{90}{\cachename{}} \\
      \midrule
      1 & 0.05 & 0.05 & \textbf{0.08} & 0.05 & 0.05 & \textbf{0.06} & 0.10 & 0.11  &   \textbf{0.19}  \\
      2 & 0.11 & 0.12  & \textbf{0.17} & 0.08 & 0.08 & \textbf{0.09} & 0.16 & 0.19 &  \textbf{0.25}  \\
      4 & 0.19 & 0.21  & \textbf{0.28} & 0.12 & 0.12 & \textbf{0.14} &0.23 & 0.26 & \textbf{0.35}  \\
      8 & 0.34 & 0.40  & \textbf{0.47} & 0.19 & 0.20 & \textbf{0.22}& 0.37 & 0.38 &   \textbf{0.44}  \\
      16 & 0.52 & 0.69  & \textbf{0.71} & 0.31 & 0.32 & \textbf{0.33} &0.47 & \textbf{0.52} & \textbf{0.52}  \\
      \bottomrule
    \end{tabular}

\end{minipage}
\hfill
\begin{minipage}[c]{0.32\linewidth}
    \centering
    \includegraphics[width=1\linewidth]{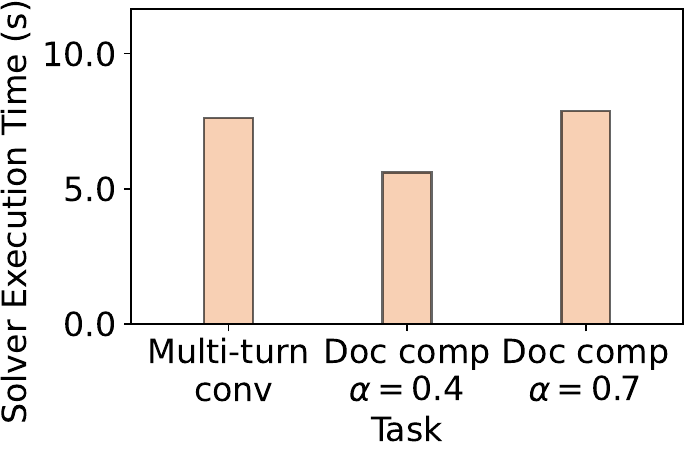}
    \caption{Constraint solver execution time for every cache resize.}
    \label{fig:ilp_runtime}
    \Description{}
\end{minipage}
\end{figure}

\subsection{Constraint Solver Overhead}\label{subsec:ilp_overhead}

We evaluate the execution time of the constraint solver when making each decision, as shown in \Cref{fig:ilp_runtime}. 
The average latency of making a cache resizing decision is as low as 7.03~s. This is a low overhead given that these serving workloads are long-running.

\begin{figure}[t]
    \begin{subfigure}[b]{1\linewidth}
    \centering
    \includegraphics[width=1\linewidth]{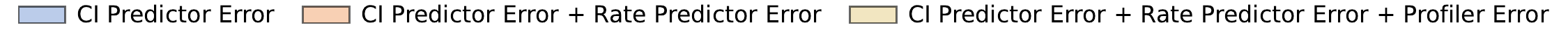}
    \end{subfigure}

    \begin{subfigure}[b]{0.32\linewidth}
    \centering
    \includegraphics[width=1\linewidth]{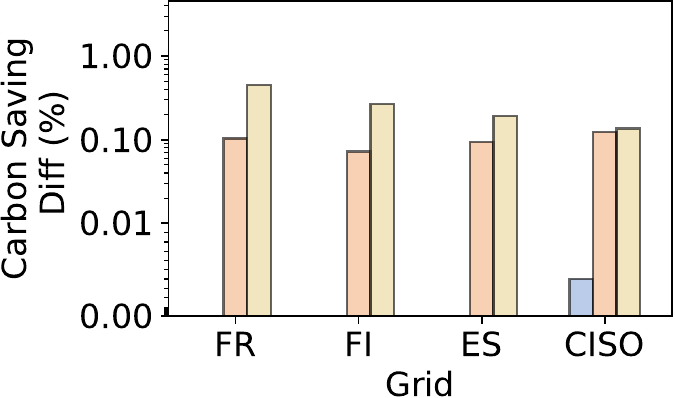}
    \caption{Multi-turn conversation.}
    \end{subfigure}
    \hfill
    \begin{subfigure}[b]{0.32\linewidth}
    \includegraphics[width=1\linewidth]{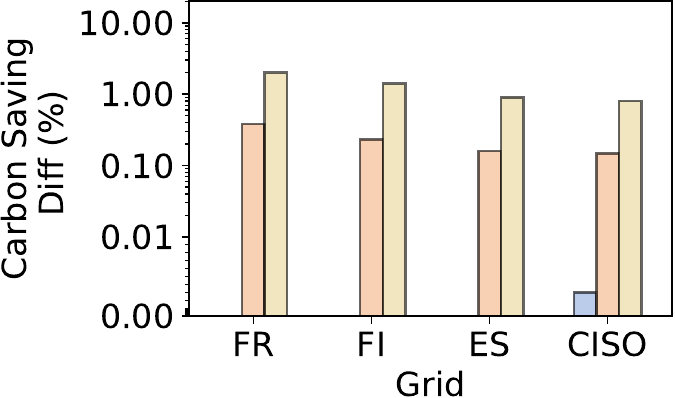}
    \caption{Document comp. ($\alpha=0.4$).}
    \end{subfigure}
    \hfill
    \begin{subfigure}[b]{0.32\linewidth}
    \includegraphics[width=1\linewidth]{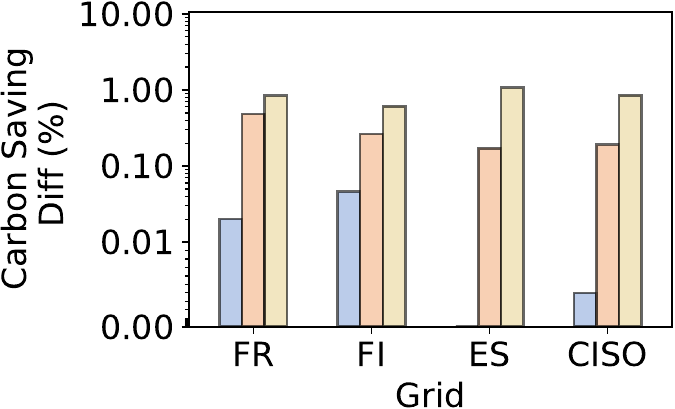}
    \caption{Document comp. ($\alpha=0.7$).}
    \end{subfigure}
\caption{Impact of prediction and profile inaccuracies (Llama-3 70B).} \label{fig:inaccuracies}
\Description{}
\end{figure}

\subsection{Predictors and Profiler Errors} \label{subsec:error_analysis}

\name{} decisions can be affected by three sources of errors, CI predictor, load predictor, and profiler as discussed in \cref{subsec:ilp_solver}. 
We first calculate the Mean Absolute Percentage Error (MAPE) of both predictors by comparing the predictions with the groundtruth. 
The load predictor has an MAPE of 4.3\,\%, and the CI predictor has MAPE values of 12.7\,\%, 15.3\,\%, 11.3\,\%, and 6.8\,\% for FR, FI, ES, and CISO grids, respectively. 
Then, we calculate the distribution differences between the profiling data and the evaluation data. 
In multi-turn conversation and document comprehension tasks, the median context length difference is 5.78\,\% and 1.13\,\%, respectively.

We further evaluate the impact on carbon emission savings by comparing the carbon emissions of \name{} under errors with an ideal scenario that uses the groundtruth. 
\Cref{fig:inaccuracies} shows the reduction of carbon emission savings due to errors. 
To make small values visible, we plot the values on a log scale. 
CI prediction errors only reduce 0.0064\,\% of carbon savings on average compared to the ideal scenario. 
When load prediction errors are included, the reduction is increased to an average of 0.20\,\%. 
Although the load predictor has lower MAPEs than CI, it directly affects the choice of cache size and thus has a higher impact on carbon emissions. 
Finally, the errors from the profiler increase the carbon savings reduction to an average of 0.79\,\%. 
The profiler errors have a relatively higher impact in low-CI grids (FR and FI), as embodied carbon has a higher weight in these grids. 
We conclude that these errors have an overall low impact on carbon emission savings.

\subsection{Sensitivity Studies}

In this section, we vary the decision-making frequency, SSD lifespan, and SSD embodied carbon to analyze their impact on carbon emissions.

\begin{figure}[t]
    \begin{subfigure}[b]{1\linewidth}
    \centering
    \includegraphics[width=1\linewidth]{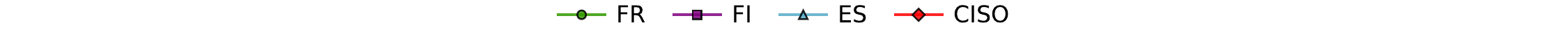}
    \end{subfigure}
    
    \begin{subfigure}[b]{0.32\linewidth}
    \centering
    \includegraphics[width=1\linewidth]{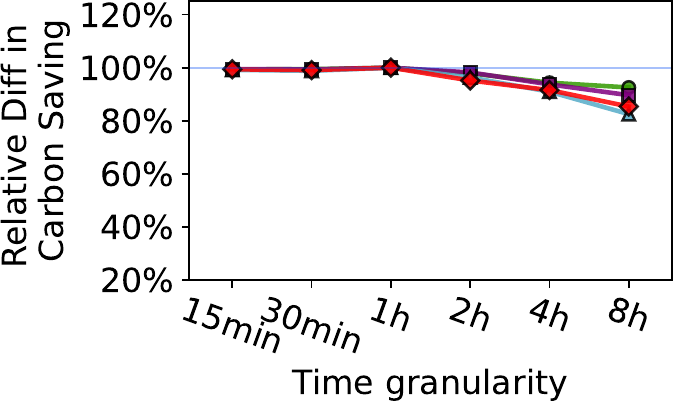}
    \caption{Multi-turn conversation.}
    \end{subfigure}
    \hfill
    \begin{subfigure}[b]{0.32\linewidth}
    \includegraphics[width=1\linewidth]{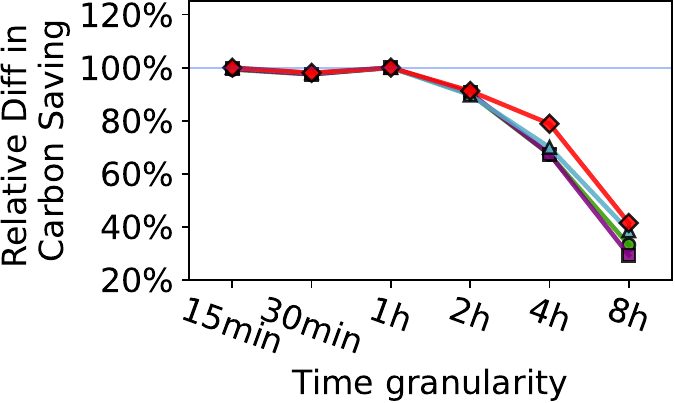}
    \caption{Document comp. ($\alpha=0.4$).}
    \end{subfigure}
    \hfill
    \begin{subfigure}[b]{0.32\linewidth}
    \includegraphics[width=1\linewidth]{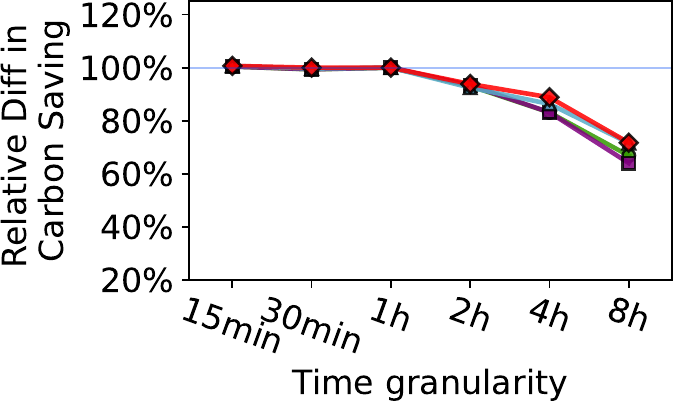}
    \caption{Document comp. ($\alpha=0.7$).}
    \end{subfigure}
\caption{Impact of variable cache resizing intervals (Llama-3 70B). A higher value indicates more savings. }
\label{fig:switching}
\Description{}
\end{figure}

\subsubsection{Cache Resizing Interval} \label{subsubsec:switch}

We evaluate the carbon emission savings of Llama-3 70B over full cache under variable resizing intervals, in addition to the default 1-hour interval. 
\Cref{fig:switching} shows the relative difference compared to the default 1-hour interval.
We also adjust the CI and rate predictors to match the granularity. 
Because the CI data has a minimum granularity of 1 hour due to the dataset, we use the fixed hourly CI value instead when the interval is within 1 hour. 
\name{} sets a sufficiently large cache size during the whole interval to ensure the SLO attainment goal. 
Therefore, the carbon emission savings are significantly reduced under longer intervals across all grids.
Among the tasks, the difference is more significant in the document comprehension that requires larger caches, especially under a lower skewness of $\alpha=0.4$.

\subsubsection{SSD Lifespan}

\begin{figure} 
\begin{subfigure}[t]{1\linewidth}
    \centering
    \includegraphics[width=0.5\linewidth]{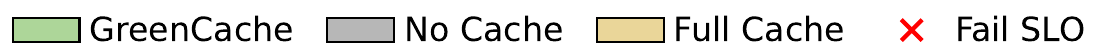}
\end{subfigure}

\begin{subfigure}[t]{0.32\linewidth}
    \centering
    \includegraphics[width=\linewidth,trim={0 0 0 16mm},clip]{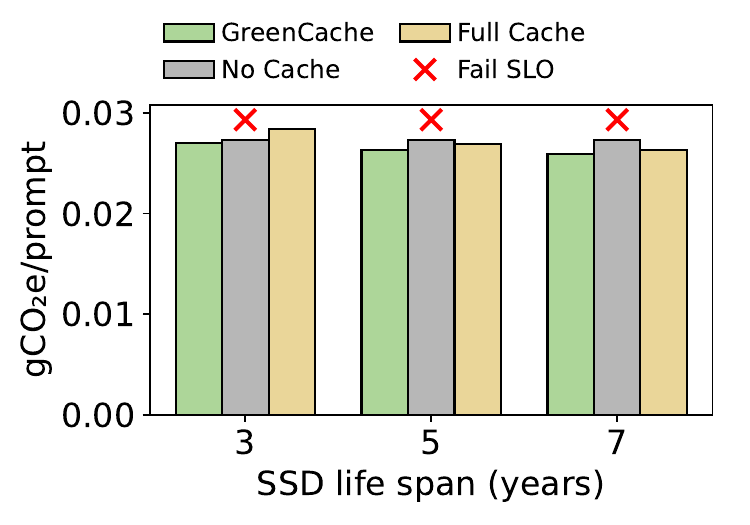}
    \caption{Multi-turn conversation.}
    \label{fig:70B_sensitivity_multiturn_life}
\end{subfigure}
\hfill
\begin{subfigure}[t]{0.32\linewidth}
    \centering
    \includegraphics[width=\linewidth,trim={0 0 0 16mm},clip]{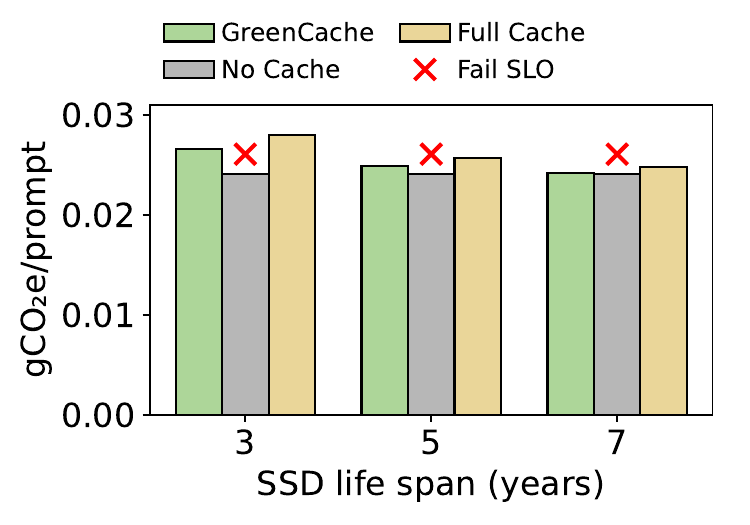}
    \caption{Document comp. ($\alpha=0.4$).}
    \label{fig:70B_sensitivity_QA0.4_life}
\end{subfigure}
\hfill
\begin{subfigure}[t]{0.32\linewidth}
    \centering
    \includegraphics[width=\linewidth,trim={0 0 0 16mm},clip]{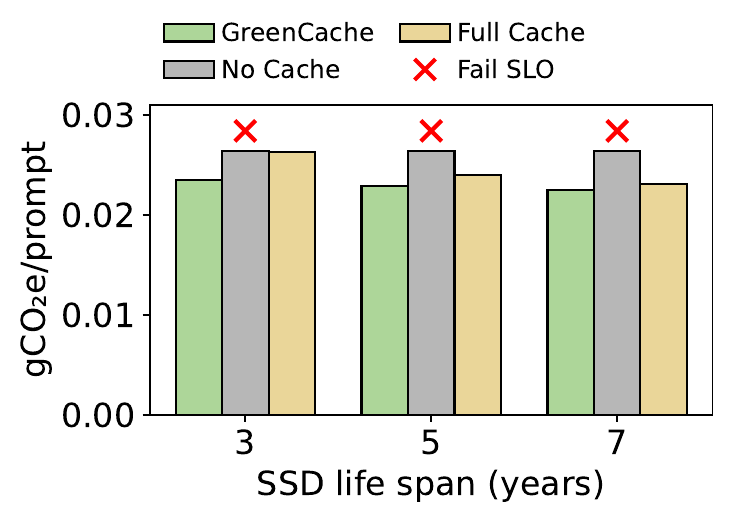}
    \caption{Document comp. ($\alpha=0.7$).}
    \label{fig:70B_sensitivity_QA0.7_life}
\end{subfigure}
\caption{Variable SSD lifespan (Llama-3 70B, ES grid). \label{fig:sensi-lifespan}}
\Description{}
\end{figure}

\begin{figure} 
\centering
\begin{subfigure}[t]{0.5\linewidth}
    \centering
    \includegraphics[width=1\linewidth]{fig/70B_sensitivity_legend.pdf}
\end{subfigure}

\begin{subfigure}[t]{0.32\linewidth}
    \centering
    \includegraphics[width=\linewidth,trim={0 0 0 16mm},clip]{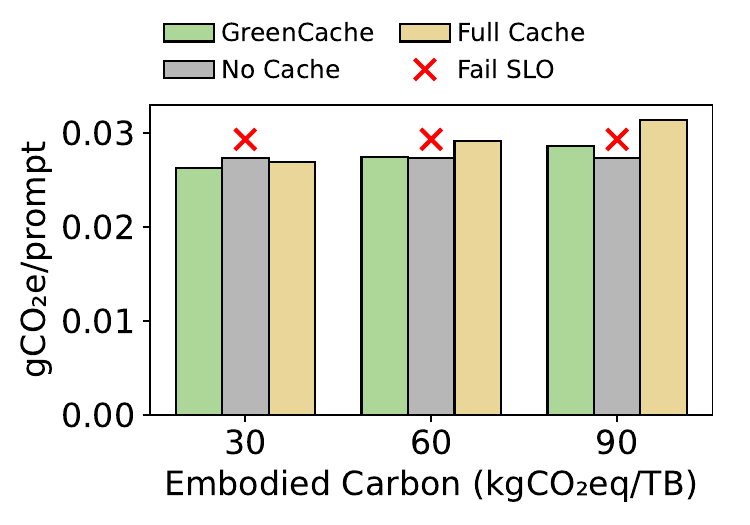}
    \caption{Multi-turn conversation.}
    \label{fig:70B_sensitivity_multiturn_embodied}
\end{subfigure}
\hfill
\begin{subfigure}[t]{0.32\linewidth}
    \centering
    \includegraphics[width=\linewidth,trim={0 0 0 16mm},clip]{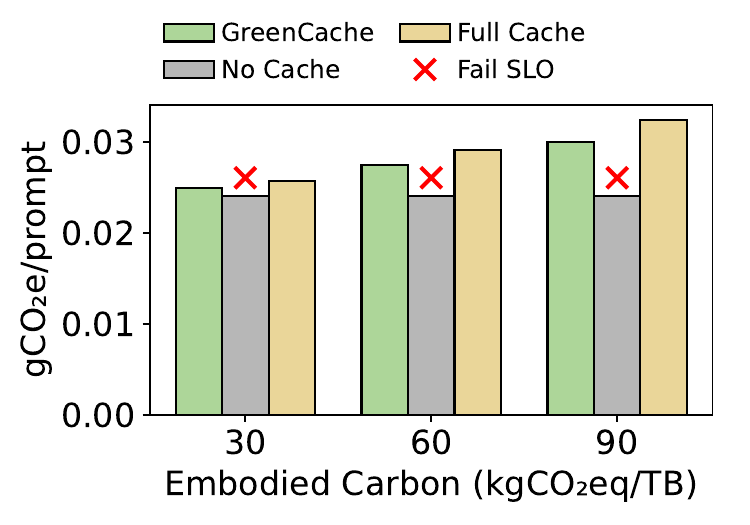}
    \caption{Document comp. ($\alpha=0.4$).}
    \label{fig:70B_sensitivity_QA0.4_embodied}
\end{subfigure}
\hfill
\begin{subfigure}[t]{0.32\linewidth}
    \centering
    \includegraphics[width=\linewidth,trim={0 0 0 16mm},clip]{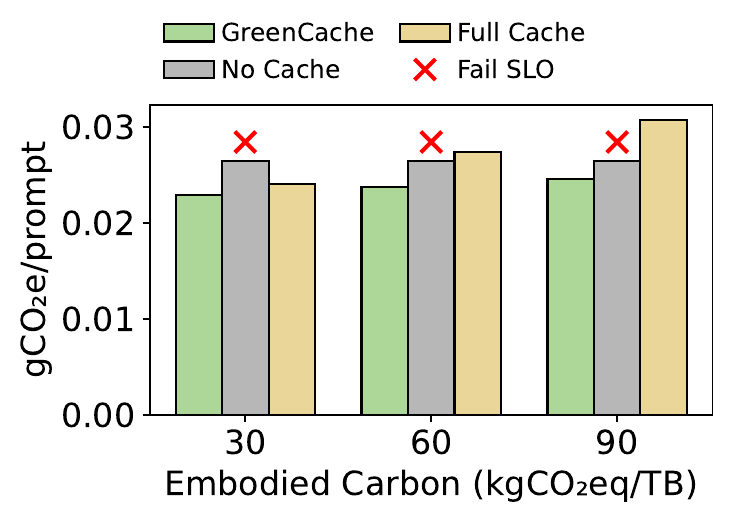}
    \caption{Document comp. ($\alpha=0.7$).}
    \label{fig:70B_sensitivity_QA0.7_embodied}
\end{subfigure}
\caption{Variable SSD embodied carbon (Llama-3 70B, ES grid).\label{fig:sensi-embodied}}
\Description{}
\end{figure}

Prior studies have found that the lifespans of SSDs vary by manufacturer and usage. In this experiment, we study the carbon emission savings from \name{} by varying the SSD lifetime from 3 to 7 years \cite{mcallister2024call,Hyrax,CarbonExplorer,lyuHotCarbon23Myths}. 
We use a fixed rate of 1.5 prompts/s for the multi-turn conversation task and 0.2 prompts/s for document comprehension. 
Like the ablation studies, we also use the average carbon intensity of the ES grid.
\Cref{fig:sensi-lifespan} indicates that a shorter, 3-year SSD lifetime leads to more carbon savings from \name{}, with up to 11.9\,\%. 
Extending SSD lifetime reduces amortized embodied carbon, lowering the savings from resizing the cache.

\subsubsection{SSD Embodied Carbon}

In this work, we model SSD embodied carbon using ACT~\cite{ACT}. 
SSD embodied carbon varies by factors such as the manufacturing technology, fab location, and account scope. Recent studies have reported higher embodied carbon of SSDs~\cite{tannu2023diretysecretssd,bhagavathula2024understanding}. 
Therefore, we evaluate a range of embodied carbon emissions of SSDs, from 30 to 90kgCO\textsubscript{2}e/TB. 
Same as the previous sensitivity study, we evaluate the two tasks with fixed rates of 1.5 and 0.2 prompts/s, respectively. 
We also follow the CI of the ES grid.
\Cref{fig:sensi-embodied} demonstrates that with the embodied carbon of SSD increasing, the carbon benefits from \name{} also increase as embodied carbon has a higher weight in the total emissions. \name{} can reduce more embodied carbon by shrinking the cache size. 
Overall, \name{} saves up to 25\,\% carbon if the SSD's embodied carbon is 90~kgCO\textsubscript{2}e/TB.

\section{Discussion}

In this section, we discuss the assumptions of embodied carbon accounting, potential improvement from more advanced load and CI predictors, and implications on MoE models.

\textbf{Embodied Carbon Accounting.} 
Consistent with prior work~\cite{han2025fair,wang2024designing,li2025ecoserve,greenLLM, greenllm_ieee_cal}, we consider both embodied and operational carbon emissions in cloud environments, emphasizing that embodied carbon is a critical component that should be included.
Different from the assumptions in Sunk Carbon Fallacy~\cite{bashir2024sunk}, which treat unselected resources as idle and their embodied carbon as a fixed sunk cost, this work assumes a cloud scenario in which unused resources are promptly utilized by other workloads. 
Consequently, we model embodied carbon as continuously amortized across workloads over time and explicitly account for it in resource provisioning and management. 
In particular, for storage that this work focuses on, \name{} dynamically resizes storage, and thus the embodied carbon of storage (\ie~SSD) is also attributed based on the actual storage provision.

\textbf{Load Prediction. }
We design a lightweight load predictor, as discussed in \Cref{subsec:load_pred}. 
The main limitation of the request rate study is the lack of real-world LLM serving traces, as the publicly available Azure trace only covers 1 week. 
With more trace data, it is possible to achieve higher prediction accuracy, like prior works in other cloud computing domains \cite{zhang2024loadpred,li2011cloudprophet,zhang2025pclrc}, which can lead to even better carbon emission savings.

\textbf{Carbon Intensity Forecast.}
\name{} uses a recent, state-of-the-art CI prediction model, EnsembleCI \cite{yan2025ensembleci}. 
CI prediction is a well-known problem in sustainable computing \cite{maji2022carboncast,maji2022dacf,yan2025ensembleci}. Like the load predictor, the CI predictor is an easily replaceable module in the \name{} framework.
With a more accurate prediction model, we expect more accurate cache allocation.

\textbf{Implications on MoE models.}
Mixture-of-Experts (MoE) models are widely deployed \cite{dai-etal-2024-deepseekmoe,geminiteam2024gemini15unlockingmultimodal}.
Compared to dense LLMs, MoE models incorporate multiple FFNs as ``experts'' and activate only a subset of them, leading to lower computation overhead. 
On the other hand, the KV cache storage is not reduced.
Therefore, MoE models lower operational carbon emissions but amplify the significance of embodied carbon, making the optimizations provided by \name{} more impactful.

\section{Related Work}

In this section, we discuss caching for LLM systems and other carbon mitigation strategies.  

\textbf{Caching for LLM.}
Storing KV caches among different requests is a common optimization \cite{cachedattention, cachegen, hcache}. 
For example, CachedAttention \cite{cachedattention} proposes to reuse the KV cache among multi-turn conversations with TBs of caches. 
CacheGen \cite{cachegen} reduces the network overhead by compressing the transferred KV caches when the SSD is deployed remotely.
HCache \cite{hcache} restores historical LLM states with intermediate activations and manages the chunk-wise storage to alleviate the overhead from I/O or recomputing. 
Instead of the precise matching in tokens, prior works also propose semantic caching. For example, GPTCache~\cite{bang2023gptcache} retrieves the cache by its semantics and returns corresponding responses. 
Unlike these studies that only target performance, \name{} achieves a tradeoff between performance and carbon emissions. 
The existing caching solutions can also be integrated into \name{} framework to reduce their total carbon emissions.

\textbf{Carbon mitigation for LLM.}
Prior works have proposed various solutions to alleviate carbon emissions of LLMs \cite{stojkovic2025dynamollm, faiz2024llmcarbon, greenLLM,greenllm_ieee_cal,samsi2023fromwordstowatt}. For example, DynamoLLM~\cite{stojkovic2025dynamollm} schedules and manages requests on various GPUs under diverse frequency settings to minimize operational carbon while satisfying SLOs.
LLMCarbon~\cite{faiz2024llmcarbon} builds a carbon model to predict the emissions due to LLM training and inference. 
GreenLLM~\cite{greenLLM, greenllm_ieee_cal} notices the heterogeneous demand in LLM and proposes to use different types of GPUs to lower both embodied and operational carbon.
From Words to Watts~\cite{samsi2023fromwordstowatt} benchmarks LLM on various settings to show the energy and performance efficiency. EcoServe~\cite{li2025ecoserve} analyzes the proportion of the operational and embodied carbon, and schedules the LLM applications with resource provision. 
These works provide substantial carbon emission savings from the computation side but overlook the importance of the embodied carbon from storage, which is \name{}'s main angle.

\section{Conclusions}
The wide use of large language models (LLMs) leads to high carbon emissions. 
While most LLM sustainability studies focus on compute-related emissions, we find that storage for saving the historical KV cache is another significant contributor.
Although caching improves performance and reduces operational emissions, it also introduces substantial embodied carbon due to the high-speed, high-capacity SSDs. To address this tradeoff, we introduce \name{}, a carbon-aware caching framework that dynamically profiles LLM tasks and uses an ILP-based optimizer to reconfigure cache size, balancing performance and carbon efficiency while meeting the SLO attainment goal.

\section*{Acknowledgement}

We thank the anonymous reviewers and the shepherd for their valuable feedback on this work, and Shuncheng Jie and Dengwang Tang for proofreading. 
We acknowledge the support of the Natural Sciences and Engineering Research Council of Canada (NSERC) RGPIN-2023-03478.
We also acknowledge the support of the U.S. National Science Foundation (NSF) CCF-2413870.


\bibliographystyle{ACM-Reference-Format}
\bibliography{bib/misc,bib/ml,bib/sys,bib/carbon}

\appendix

\section{ILP Complexity} \label{sec:ilp_comp}

\begin{theorem*}[NP-hardness]
The GreenCache optimization problem in Eq.~(\ref{eq:ilp}) is NP-hard, even in a restricted setting where each time step only
allows a binary cache decision (off/on), and the SLO requirement is a global ratio constraint:
at least a configurable fraction $\rho$ of all requests over the horizon satisfy the TTFT threshold and at least
the same fraction $\rho$ satisfy the TPOT threshold (where $\rho$ is part of the input).
\end{theorem*}

\begin{proof}
We prove NP-hardness by a polynomial-time reduction from the NP-complete \textsc{0--1 Knapsack} decision problem.
A \textsc{Knapsack} instance consists of $m$ items, where item $k$ has weight $w_k \in \mathbb{Z}_{>0}$ and value
$v_k \in \mathbb{Z}_{>0}$, a weight budget $W$, and a target value $V$.
The question is whether there exists a subset $K \subseteq \{1,\dots,m\}$ such that
$\sum_{k\in K} w_k \le W$ and $\sum_{k\in K} v_k \ge V$.

\par\addvspace{0.25\baselineskip}
\noindent\textbf{Decision version of GreenCache.}
Consider the decision form: given a carbon budget $C$, is there a cache plan such that total carbon $\le C$ and
the (global-$\rho$) SLO constraints hold? 

\par\addvspace{0.25\baselineskip}
\noindent\textbf{Construction.}
Given a \textsc{Knapsack} instance, we construct a restricted GreenCache instance as follows.

\begin{itemize}[leftmargin=*]
  \item 
  Set $T=m$, which maps the number of items ($m$) in the Knapsack problem directly to the total number of time steps ($T$) in the GreenCache instance. Map item $k$ to time step $t=k$.
  Restrict the cache decision at each step to $S_k \in \{0,1\}$ (cache off/on).

  \item   The Knapsack target value $V$ maps to the minimum number of requests that must satisfy the SLO constraints in the GreenCache instance.

  \item 
  Set the request volume at time step $k$ to be $\lambda_k := v_k$.
  Let $\Lambda := \sum_{k=1}^m \lambda_k = \sum_{k=1}^m v_k$.
  If $V > \Lambda$, the knapsack instance is trivially infeasible; in this case we output any GreenCache instance that is
  trivially infeasible (e.g., by setting an impossible SLO), and the reduction remains correct. Hence, assume $V \le \Lambda$.

  \item 
  Fix the workload SLO thresholds $\theta_{\text{TTFT}}$ and $\theta_{\text{TPOT}}$.
  For each step $k$, let $q^{\text{TTFT}}_k$ (resp.\ $q^{\text{TPOT}}_k$) denote the number of requests at step $k$
  whose TTFT (resp.\ TPOT) meets the corresponding threshold.
  We enforce that at least a fraction $\rho$ of all requests satisfy each threshold:
  \[
     \sum_{k=1}^m q^{\text{TTFT}}_k \ge \rho \Lambda,\qquad
     \sum_{k=1}^m q^{\text{TPOT}}_k \ge \rho \Lambda.
  \]
  We set the (input) ratio parameter to
  \[
     \rho := \frac{V}{\Lambda}\in(0,1].
  \]
  Hence $\rho\Lambda = V$, and the global-$\rho$ constraints become
  $\sum_k q^{\text{TTFT}}_k \ge V$ and $\sum_k q^{\text{TPOT}}_k \ge V$.

  \item 
  We choose the instance's latency functions/coefficients so that, for each step $k$:
  \begin{enumerate}
    \item If $S_k=1$, then all $\lambda_k$ requests have TTFT and TPOT below the corresponding thresholds, i.e.,
          $q^{\text{TTFT}}_k = q^{\text{TPOT}}_k = \lambda_k$.
    \item If $S_k=0$, then all $\lambda_k$ requests violate at least one of the two thresholds, so $q^{\text{TTFT}}_k = q^{\text{TPOT}}_k = 0$.
  \end{enumerate}
  Therefore, the number of requests satisfying each threshold at step $k$ is $q^{\text{TTFT}}_k = q^{\text{TPOT}}_k = \lambda_k S_k$.

  \item 
  Set the incremental carbon at time step $k$ to be exactly $w_k \cdot S_k$,
  and set the carbon budget to $C := W$.
  Any carbon terms that are independent of the cache decisions (e.g., constant embodied carbons) are
  in the budget $C$ without affecting NP-hardness.
\end{itemize}

\par\addvspace{0.25\baselineskip}
\noindent\textbf{Correctness.}
By construction, for every $k$ we have $q^{\text{TTFT}}_k = q^{\text{TPOT}}_k = \lambda_k S_k$.
Thus, the global SLO constraints become
\[
\sum_{k=1}^m \lambda_k S_k \ge V
\quad\Longleftrightarrow\quad
\sum_{k=1}^m v_k S_k \ge V.
\]
Meanwhile, the carbon budget constraint is exactly
\[
\sum_{k=1}^m w_k S_k \le W.
\]
Therefore, there exists a feasible solution to the constructed GreenCache decision instance iff there exists a feasible
knapsack subset.

\par\addvspace{0.25\baselineskip}
\noindent\textbf{Polynomial-time reduction.}
The construction sets $T=m$ and uses only $\{(w_k,v_k)\}_{k=1}^m$ and simple arithmetic. The rational parameter
$\rho = V/\Lambda$ has polynomial bit-length (it can be represented by the integer pair $(V,\Lambda)$).
Hence the reduction runs in polynomial time.

\par\addvspace{0.25\baselineskip}
Therefore, the GreenCache problem is NP-hard even in this restricted special case; consequently, the general
optimization problem in \Cref{eq:ilp} is NP-hard.
\end{proof}


\end{document}